\documentclass{amsart}
\usepackage{amssymb}
\usepackage{amsmath}
\usepackage{amsfonts}
\usepackage{amsthm}
\usepackage{color}
\usepackage{url}
\usepackage[normalem]{ulem} 
\newcommand{\Ac}{\mathcal{A}}
\newcommand{\Bc}{\mathcal{B}}
\newcommand{\Cc}{\mathcal{C}}
\newcommand{\Dc}{\mathcal{D}}
\newcommand{\Ic}{\mathcal{I}}
\newcommand{\Jc}{\mathcal{J}}
\renewcommand{\Mc}{\mathcal{M}}
\newcommand{\Oc}{\mathcal{O}}
\newcommand{\Vc}{\mathcal{V}} 

\newcommand{\FF}{\mathbb{F}}
\newcommand{\QQ}{\mathbb{Q}}

\newcommand{\ZZ}{\mathbb{Z}}

\newcommand{\qf}{\mathfrak{q}}
\newcommand{\Qf}{\mathfrak{Q}}

\def\MBZ{\mathbb{Z}}

\def\MBQ{\mathbb{Q}}

\def\Z4{\MBZ/4\MBZ}

\def\MCO{\mathcal{O}}

\def\LA{\Lambda}

\def\sigmabar{\overline{\sigma}}
\def\ubar{\overline{u}}
\def\Kibar{\overline{K}^{(i)}}

\def\K1bar{\overline{K}^{(1)}}
\def\Krbar{\overline{K}^{(g)}}
\def\Kbar{\overline{K}}
\def\Fbar{\overline{F}}

\newtheorem{lemma}{Lemma}
\newtheorem{proposition}{Proposition}

\newtheorem{definition}{Definition}

\theoremstyle{definition}
\newtheorem{example}{Example}

\newtheorem{remark}{Remark}

\begin{document}
\title {Quotients of Orders in Cyclic Algebras and Space-Time Codes}
\author{F. Oggier}
\address{Division of Mathematical Sciences, School of Physical and Mathematical Sciences, \newline\indent Nanyang Technological University, Singapore.} 
\email {frederique@ntu.edu.sg}

\author{B. A. Sethuraman}
\address{Department of Mathematics, California State University Northridge,  \newline \indent
Northridge, CA 91330, USA.}
\email{al.sethuraman@csun.edu}


\thanks {F. Oggier was partially supported by the Singapore National Research Foundation under 
Research grant NRF-CRP2-2007-03.  B.A. Sethuraman was partially supported by the National Science Foundation under Research grant DMS-0700904.  
  This paper had its genesis in some preliminary results that were obtained when both authors were visiting the Indian Institute of Technology Bombay, and both authors are grateful to the mathematics department there for its warm hospitality.}
\begin{abstract}
Let $F$ be a number field with ring of integers $\Oc_F$ and $\Dc$ a division $F$-algebra with a maximal cyclic subfield $K$. We study rings occurring as quotients of a natural $\Oc_F$-order $\Lambda$ in $\Dc$ by  two-sided ideals. We reduce the problem to studying the ideal structure of $\Lambda/\qf^s\Lambda$,  where $\qf$ is a prime ideal in $\Oc_F$, $s\geq 1$. We study the case where $\qf$ remains unramified in $K$, both when $s=1$ and $s>1$. This work is motivated by its applications to space-time coded modulation.
\end{abstract}
\maketitle

%
%
\section{Introduction} \label{sec:intro}

Space-time coding is a wireless coding technique to ensure reliability when multiple antennas are in use, both at the transmitter and receiver ends of a communication channel.  
A space-time code typically consists of a family $\Cc$ of $n\times n$ complex matrices, where $n$ is the number of transmit antennas. The two main design criteria are 
the {\em rank criterion}, asking that
\[
\det(X-X')\neq 0,~X\neq X' \in \Cc,
\]
and the {\em minimum determinant criterion}, which further requires to maximize
\[
\min_{X\neq X' \in \Cc} |\det(X-X')|^2.
\]
Cyclic division algebras are by now a well known tool to design space-time codes \cite{Sethuraman,survey}. 
A cyclic algebra is embedded into a ring of $n\times n$ matrices by right multiplication, and the resulting codebook $\Cc$ has the property that $X-X'\in\Cc$, which simplifies the rank criterion to
\begin{equation}\label{eq:div}
\det(X)\neq 0,~0\neq X \in \Cc,
\end{equation}
and the minimum determinant criterion to maximizing
\begin{equation}\label{eq:det}
\min_{0\neq X \in \Cc} |\det(X)|^2.
\end{equation}
Furthermore, a cyclic algebra which is division satisfies the property that any nonzero element is invertible, which immediately fulfills the rank criterion (\ref{eq:div}).

A variation of this problem is that of space-time coded modulation for slow multiple 
antenna fading channels, where we assume that the channel is constant over $nL$ channels uses, namely, the code $\Cc$ now contains $n\times nL$ codewords of the form $X=(X_1,\ldots,X_L)$, where every $X_i$ is an $n\times n$ complex matrix.
The rank criterion then requires 
$\det(XX^*)=\det(X_1X_1^*+\ldots+X_LX_L^*) \neq 0$ 
when $X\neq 0$, which can be achieved by using $L$ $n\times n$ independent space-time block codes satisfying (\ref{eq:div}). However, a better performance regarding the second criterion,  {that is} maximizing 
\begin{equation}\label{eq:mindet}
\Delta_{min}=\min_{0\neq X \in \Cc} |\det(XX^*)|
\end{equation}
can be obtained using coset coding \cite{Luzzi,pstm}. 

Coset codes are constructed using the following procedure:  One starts with a standard space-time code, which is an order $\LA$ inside a cyclic division algebra over a number field.  This is known as the inner code. One takes the quotient of this order by a two-sided ideal $\Jc$ and obtains a finite ring $\LA/\Jc$, over which one constructs a code $\bar{\Cc}$ of length $L$. This code is referred to as the outer code, and it is an additive subgroup of $\oplus_{i=1}^{L}\LA/\Jc$. Writing $\pi$ for the induced map $\oplus_{i=1}^{L}\LA \rightarrow \oplus_{i=1}^{L}\LA/\Jc$, one then chooses as the final code $\Cc$ an additive subgroup of $\oplus_{i=1}^{L}\LA$ such that $\pi(\Cc) = \bar{\Cc}$. This is hence of length $L$, and every of the $L$ coefficients is contained in the inner code.  The challenge in this procedure is to design $\bar{\Cc}$ so as to maximize $\Delta_{min}$.  

Previous works \cite{Luzzi,pstm} focus on coset codes obtained by choosing an order corresponding to one of the perfect codes \cite{perfect} of size $2$, $3$, or $4$, and selecting specific two-sided ideals  that yield as quotients the matrix rings $\Mc_m(S)$, where $m$ is respectively 2, 3 or 4, and $S$ is either a finite field or a finite extension ring of a finite field.

\subsection{Contribution and Organization} \label{subsecn:contrib}

Our goal is to study coset coding in cyclic division algebras in greater generality than considered previously. We will focus in this paper on describing various families of rings that can occur as quotients of orders in division algebras over number fields. We  present some new coset encoding schemes suitable 
for space-time coded modulation. The known results of \cite{Luzzi,pstm} are particular cases of our results. 

We start in Section \ref{sec:cyc} by recalling how space-time codes are obtained from a division algebra $\Dc$ with center a number field $F$ and with a given maximal cyclic subfield $K$. 
In this same section we introduce the natural order $\Lambda$ of $\Dc$ that will be the focus of this paper, and we describe the problem of designing coset codes using $\Lambda$.  We establish a key inequality concerning  the determinant of a sum of positive-definite matrices that will be needed for coding gain estimates.
In Section \ref{sec:prel} we prove an easy but fundamental lemma that describes the  quotients of $\Lambda$ by a two-sided ideal as direct sums of (generalized) cyclic algebras over finite rings. This lemma shows that it is enough to study the case where this two-sided ideal is of the form  $\qf^s\Lambda$ where $\qf$ is a prime ideal of the ring  of integers $\Oc_F$ of $F$, $s\ge 1$.  In the remaining sections, we study the various cases where $\qf$ remains unramified in $K$.
 We show that in several of these cases, these generalized cyclic algebras over finite rings can  alternatively be described in terms of better known rings, and we describe the ideals of these algebras in such cases. We give examples of codes to illustrate the different quotient rings we obtain.
Finally, in Section \ref{sec:wiretap} we summarize the implications of our results for space-time coded modulation, and describe another potential application to wire-tap coding.   

The results of this paper identify a large class of rings as possible quotient rings of $\LA$, and thus provide a framework for coset coding in division algebras.
%
%
\section{Coset Coding and Cyclic Algebras}
\label{sec:cyc}

Let $K/F$ be a cyclic number field extension of degree $n$, with cyclic Galois group $G=\langle\sigma 
\rangle$. Having coding applications in mind, we will typically be interested in having $F = \MBQ(\imath)$, $\imath^2=-1$, $F = \MBQ(\omega_3)$, $\omega_3^2+\omega_3+1=0$, or $F = \MBQ$, although we will not make any such restrictions in our results.
\begin{definition}
A cyclic algebra $\Ac=(K/F,\sigma,u)$ is a vector space 
\[
K\oplus Kz \oplus \cdots \oplus Kz^{n-1}
\]
with multiplication defined by $zk = \sigma(k)z$ for all $k\in K$, and $z^n = u\in F$, $u\neq 0$.
\end{definition}
A cyclic algebra is said to be a cyclic {\em division} algebra if every nonzero element is invertible, in which case we will use the notation $\Dc = (K/F, \sigma, u)$.

A space-time codebook is obtained via the standard $F$-embedding of $\Dc$ into $\Mc_n(K)$ arising from representing the action of $\Dc$ by right multiplication on $\Dc$ viewed as a left $K$-space, in the basis $\{1,z,\dots,z^{n-1}\}$. For $x_i=x_{i,0}+x_{i,1}z+\ldots+x_{i,n-1}z^{n-1}\in \Dc$, write $M(x_i)$ for its image: 
\begin{equation}\label{eq:mat} 
M(x_i)= 
\begin{bmatrix}
x_{i,0}    & u \sigma(x_{i,n-1}) & u \sigma^2(x_{i,n-2}) &\dots   & u \sigma^{n-1}(x_{i,1}) \\
x_{i,1}    & \sigma(x_{i,0})     & u \sigma^2( x_{i,n-1}) & \dots  & u \sigma^{n-1}(x_{i,2}) \\
\vdots & \vdots&  \vdots &\ddots &\vdots\\
\vdots      &  \vdots          &    \vdots         &  &  u \sigma^{n-1}(x_{i,n-1})      \\
x_{i,n-1}  & \sigma(x_{i,n-2})   & \sigma^2(x_{i,n-3})   &   & \sigma^{n-1}(x_{i,_0}) \\
\end{bmatrix}.
\end{equation} 
Let $\Oc_F$ and $\Oc_K$ be the ring of integers of $F$, respectively $K$. The following will be fundamental to us:

\par  \textit{We  will assume throughout the paper that $u \in \MCO_F$} (instead of merely taking $u$ to be in $F$). 

Recall (see, e.g., \cite[Chap. 2, \S 8]{Reiner}) that for $R$ a Noetherian integral domain with quotient field $F$, and $\Ac$ a
finite dimensional $F$-algebra, we have the following definition. 

\begin{definition}\label{def:order}
An $R$-{\em order} in the $F$-algebra $\Ac$ is a subring $\Lambda$ of $\Ac$,
having the same identity element as $\Ac$, and such that $\Lambda$ is
a finitely generated module over $R$ and generates $\Ac$ as a linear space
over $F$. An order $\Lambda$ is called {\em maximal} if it is not properly
contained in any other $R$-order.
\end{definition} 

We consider the subring $\LA = (\MCO_K/\MCO_F, \sigma, u)$ of the cyclic division $F$-algebra $\Dc= (K/F, \sigma, u)$:
\[
\LA=\MCO_K \oplus \MCO_K z \oplus \ldots \oplus \MCO_K z^{n-1}
\]
(recall the assumption that $u\in \MCO_F$).  This is an $\MCO_F$-order in $\Dc$, and we will refer to this order in the paper as  {\em the natural order}. (Of course, choosing a different cyclic maximal subfield $K'$ will yield a different representation of $\Dc$ as $(K'/F,\sigma,u')$, and will lead to a different order. Thus, we assume at the outset that the representation $\Dc = (K/F, \sigma, u)$ is fixed and define the natural order with respect to this representation.)  We note that $\Lambda$ is in general not a maximal order, but this will not concern us in the paper.

Let $\Jc$ be a two-sided ideal of $\Lambda$. 
As described in Section \ref{sec:intro}, coset coding consists of designing a code $\bar{\Cc}$ over the quotient ring $\Lambda/\Jc$. The space-time codebook $\Cc$ is then obtained as a lift of $\bar{\Cc}$.
In order to perform coset coding, we therefore first need  to determine  the structure of $\Lambda/\Jc$ when $\Lambda=\oplus_{i=0}^{n-1}\Oc_K z^i$ and $\Jc$ is a two-sided ideal of $\Lambda$.

The code design criterion for $\bar{\Cc}$ is not obvious. The code $\bar{\Cc}$ must be designed so that $\Delta_{min}$ is maximized for $\Cc$. We first establish a critical  inequality that we will need  in (\ref{eqn:det_stuff}) below.  This inequality is proved in \cite[Lemma 1]{Luzzi} for the special case of  $2\times 2$ matrices coming from the Golden Code, but as we see below, it holds more generally for $k$-tuples of $n\times n$ positive-definite matrices.  It is presumably well-known, but we provide a proof for convenience. Note that any positive definite (necessarily Hermitian) matrix can be written as $XX^*$ for some  invertible matrix $X$, not necessarily unique, and note that the sum of two positive definite matrices as well as the inverse of a positive definite matrix are also positive definite.

\begin{lemma} \label{det_inequal} Let $Y_1$, $\dots$, $Y_k$ be $n\times n$ positive definite  matrices, and write $Y_1 + \cdots + Y_k$ as $S S^* $ for a suitable invertible matrix $S$.  If each $Y_i$ equals $S_iS_i^* $ for suitable invertible matrices $S_i$, then $|\det(S)| \ge  \sum_{i=1}^k |\det(S_i)|$.
In particular, if $X_1$, $\dots$, $X_k$ are $n\times n$ invertible matrices with complex entries, then $|\det(\sum_{i=1}^n X_i X_i^*)| \ge \left(\sum_{i=1}^k |\det(X_i)|\right)^2$.
\end{lemma}

\begin{proof} By induction, it is sufficient to prove this in the case $k=2$. Write $Y$ for $Y_1+ Y_2$.  Since $|\det(S)|^2 = \det(Y)$, and similarly for $|\det(S_i)|$, we need to show that $(\det(Y))^{1/2} \ge (\det(Y_1))^{1/2}  + (\det(Y_2))^{1/2} $. Writing the left side as $(\det(Y_1))^{1/2} (\det(I_n + Y_1^{-1}Y_2))^{1/2}$, where $I_n$ is the identity $n\times n$ matrix, it is sufficient to show that $(\det(I_n + Y_1^{-1}Y_2))^{1/2} \ge 1 + (\det(Y_1^{-1}Y_2))^{1/2}$.  Write $C$ for $Y_1^{-1}Y_2$.  By \cite[Theorem 7.6.3]{HJ}, $C$ is diagonalizable and has all positive eigenvalues (since $Y_2$ has all positive eigenvalues). Therefore, there exists some invertible matrix $P$ such that $C = PDP^{-1}$, where $D$ is diagonal, with positive entries $l_1$, $\dots$, $l_n$.  Then $(\det(I_n + C))^{1/2} = (\det(P(I_n+D)P^{-1}))^{1/2} = (\det(I_n + D))^{1/2} = \sqrt{(1+l_1)\cdots(1+l_n)}$.
 We thus need to show that $\sqrt{(1+l_1)\cdots(1+l_n)} \ge 1 + \sqrt{l_1\cdots l_n}$, or what is the same thing, ${(1+l_1)\cdots(1+l_n)} \ge \left(1 + \sqrt{l_1\cdots l_n}\right)^2$.  The left side contains as summands $1$, $(l_1\cdots l_n)$, $l_1$, $(l_2\cdots l_n)$, and (for $n\ge 3$) other positive quantities. The right side equals $1 + (l_1\cdots l_n) + 2\sqrt{l_1\cdots l_n}$.  By the arithmetic-geometric mean inequality, $l_1 + (l_2\cdots l_n) \ge 2\sqrt{l_1\cdots l_n}$, which proves the inequality.

The last statement follows by taking $Y_i = X_i X_i^*$, $S_i = X_i$, and noting that the $Y_i$ are positive definite as the $X_i$ are invertible (see \cite[Theorem 7.1.6]{HJ} for instance).
\end{proof}

(We thank Chi-Kwong Li for pointing out that this inequality can also be derived from 
Minkowski's inequality,
\cite[Theorem 7.8.8]{HJ}. Indeed our proof above is similar to the proof of Minkowski's inequality.)

Recalling from (\ref{eq:mindet}) the definition of $\Delta_{min}$, we observe that
\begin{eqnarray} 
\Delta_{min}&=&\min_{0\neq X \in \Cc} |\det(XX^*)| \notag\\
            &=& \min_{0\neq X \in \Cc} |\det(X_1X_1^*+\ldots X_LX_L^*)| \notag\\
            &\geq & \min_{0\neq X \in \Cc} \left( \sum_{i=1}^L |\det(X_i)|  \right)^2, \label{eqn:det_stuff}
\end{eqnarray} 
where the last inequality comes from  Lemma \ref{det_inequal} above. Note that this lower bound also proves that chosing $X_1,\ldots,X_L$ independently will make sure that $\det(XX^*)\neq 0$, as claimed in the introduction.
Now in this paper $X_i=M(x_i)$, where $M(x_i)$ is defined in (\ref{eq:mat}). Denote by $\bar{x}_i$ the image of $x_i$ in $\Lambda/\Jc$, and 
by $\bar{x}$ the codeword $(\bar{x}_1,\ldots,\bar{x}_L)\in \bar{\Cc}$.  
If $\bar{x}\neq 0$ (so some $\bar{x_i} \neq 0$, i.e., some $x_i \not\in\Jc$), then the number of nonzero components of $\bar{x}$ is at least the Hamming distance $d_H(\bar{\Cc})$ of $\bar{\Cc}$, and 
\[
\Delta_{min}\geq  d_H(\bar{\Cc})^2 \min_{{x_i}\not\in\Jc} |\det(X_i)|^2.
\]   
Now if $\bar{x}=0$, then every $x_i$ is {either a non-zero element of $\Jc$ or $0$}, and 
\[
\Delta_{min}\geq  \min_{ 0 \neq x_i \in \Jc} |\det(X_i)|^2.
\]
Thus 
\begin{equation}\label{eq:dhbound}
\Delta_{min} \geq \min\left( d_H(\bar{\Cc})^2 \min_{{x_i}\not\in\Jc} |\det(X_i)|^2, \min_{0\neq x_i \in \Jc} |\det(X_i)|^2\right). 
\end{equation} 
To further analyze this last expression, we need to understand $\Lambda/\Jc$ well. We will prove a structure theorem for $\Lambda/\Jc$  below and analyze the structure in greater details in different cases. Particular emphasis will be given in the examples to quotient rings of characteristic 2, more suitable for coding applications.

%
%
%
\section{Quotients of The Natural Order: General Structure} 
\label{sec:prel}

We prove in this section an elementary but fundamental lemma that describes the quotients of  $\Lambda=\oplus_{i=0}^{n-1}\Oc_K z^i$ in terms of generalized cyclic algebras over finite rings.

We start with some preliminary results on ideals of $\Lambda$.

\begin{lemma} \label{ring_is_integral} The elements of $\LA$ are integral over $\MCO_F$.
\end{lemma}

\begin{proof} We consider the image $M(l)$ of $l\in \LA$ by the standard $F$-embedding of $\Dc$ into $\Mc_n(K)$, as described in (\ref{eq:mat}). Then $M(l)$ satisfies its own characteristic polynomial by Cayley-Hamilton {Theorem}, and since this is a polynomial over $F$ and the embedding is an $F$-embedding, $l$ also satisfies this polynomial.  But the entries of $M(l)$ are all in $\MCO_K$ (note that $u\in \MCO_F\subseteq \MCO_K$).  It follows that the characteristic polynomial (which is monic) has its coefficients in $\MCO_F$.  Thus, $l$ is integral over $\MCO_F$.
\end{proof}

Our next result, whose proof is easy, enables us to study ideals of $\LA$ (and hence quotients of $\LA$) in terms of ideals of $\MCO_F$.

\begin{lemma} \label{cap_nonzero}
Let $\Jc$ be a nonzero (two-sided) ideal of $\LA$. Then $\Jc \cap \Oc_F  \neq 0$.
\end{lemma}

\begin{proof} Pick $j\in \Jc$, $j\neq 0$. Then $j$ is integral over $\MCO_F$ by Lemma \ref{ring_is_integral}, so $j^t + f_{t-1} j^{t-1} + \cdots + f_1j + f_0 =0$ for some $f_i\in \MCO_F$.  We may assume $f_0 \neq 0$ since $\LA$ is a domain, namely, if $f_0, \dots, f_{i-1}=0$ but $f_i \neq 0$, then we could factorize $j^t + f_{t-1} j^{t-1} + \cdots + f_ij^i=j^i(j^{t-i} +  {f_{t-1}j^{t-i-1}} + \cdots + f_i)$ and since 
$\Lambda$ has no zero divisors, this would give a new polynomial for $j$ with nonzero constant term. Then $f_0 = -(j^t + f_{t-1} j^{t-1} + \cdots + f_1j)$ is in $\Jc\cap\MCO_F$.
\end{proof}

This lemma gives us a strategy for identifying every nonzero two-sided ideal $\Jc$ of $\LA$ (and hence for identifying the quotient ring $\Lambda/\Jc$) by focusing first on nonzero ideals of the nicer ring $\Oc_F$.  The lemma shows that for every two-sided ideal $\Jc$ of $\Lambda$, the intersection
$
\Ic=\Jc\cap \MCO_F
$ is a nonzero ideal of $\Oc_F$. On the other hand, if $\Ic$ is any nonzero ideal of $\Oc_F$, it lies in the center of $\LA$, and generates the two-sided ideal $\Ic\LA $ of $\LA$ given by 
\begin{equation}  \label{eqn:ILA_desc}
\Ic\LA = \{i l \ | \ i\in \Ic, l\in \LA\} = \{i_0 + i_1 z + \cdots + i_{n-1}z^{n-1}, \ i_0, \dots, i_{n-1}\in\Ic\Oc_K \}.
\end{equation}
The second description above shows that $\Ic\LA \cap \MCO_F = \Ic$.  We thus have a \textit{surjective} (i.e., \textit{onto}) map $\Phi$ from the set of nonzero ideals of $\LA$ to the set of nonzero ideals of $\Oc_F$ that takes the ideal $\Jc$ of $\LA$ to $\Phi(\Jc)=\Ic=\Jc\cap\Oc_F$.  Note that if $\Phi(\Jc) = \Ic$, then $\Jc$ must contain the ideal $\Ic\LA$ generated by $\Ic$.
The map $\Phi$ is not 1-1, since many ideals of $\Lambda$ could intersect down to the same ideal of $\Oc_F$. However, the surjectivity of $\Phi$ shows that if we vary $\Ic$ in $\Oc_F$ and for each $\Ic$ if we consider all ideals $\Jc$ of $\LA$ such that $\Jc\cap\Oc_F = \Ic$, we would identify all ideals of $\LA$.  But we could be even looser: for each $\Ic$, we could consider all ideals $\Jc$ of $\LA$ that \textit{merely contain} $\Ic$.  These would of course include all ideals in $\Phi^{-1}(\Ic)$, but could contain more.   The advantage of this approach is the following: Notice that $\Jc \supseteq \Ic$ if and only if $\Jc\supseteq \Ic\LA$. Hence, we could appeal to standard homomorphism theorems for rings:
there is a one-to-one correspondence between ideals of $\Lambda$ that contain $\Ic\Lambda$ and ideals of the quotient $\Lambda/\Ic\Lambda$, with an ideal $\Jc \supseteq \Ic\Lambda$ of $\Lambda$ corresponding to the ideal $\Jc/\Ic\Lambda$ of $\Lambda/\Ic\Lambda$.  Therefore, to determine all ideals $\Jc$ of $\LA$ (and hence to determine all  quotient rings $\LA/\Jc$), it is sufficient to determine the ideal structure of quotient rings of the special form $\LA/\Ic\LA$, where $\Ic$ is a nonzero ideal of $\Oc_F$.


Notice from the discussion above that $\Oc_F/\Ic$ injects into $\LA/\Ic\LA$ (since $\Ic\LA\cap\Oc_F = \Ic$), and that 
\begin{equation} \label{eq:direct_sum}
\LA/\Ic\LA \cong \oplus_{i=0}^{n-1} (\MCO_K/\Ic\MCO_K)z^i
\end{equation}
where $z(k+\Ic\MCO_K) = (\sigma(k)+\Ic\MCO_K)z$ for all {$k\in \Oc_K$} and $z^n = u+\Ic$.  
(Note that  $G = \langle\sigma\rangle$ acts trivially on $\Ic\subseteq \MCO_F$, and hence, 
$G$ has a natural action on $\MCO_K/{\Ic}\MCO_K $.) 

Using that $\MCO_F$ is a Dedekind domain, $\Ic$ can be factored as $\qf_1^{s_1}\cdots \qf_t^{s_t}$, where the $\qf_i$ are prime (and maximal) ideals of $\Oc_F$. Since the ideals $\qf_i^{s_i}$ are pairwise comaximal already as ideals of $\MCO_F$,  the Chinese Remainder theorem allows us to write $\MCO_F/\Ic$ and $\MCO_K/\Ic\MCO_K$ as ring direct sums:
\begin{eqnarray} \label{CRTIsoms1}
\MCO_F/\Ic &=& \MCO_F/\qf_1^{s_1}\cdots \qf_t^{s_t}  \cong \MCO_F/\qf_1^{s_1} \times \cdots \times \MCO_F/\qf_t^{s_t}\notag\\
\label{CRTIsoms2}
 \MCO_K/\Ic\MCO_K &=& \MCO_K/\qf_1^{s_1}\cdots \qf_t^{s_t}\MCO_K \cong \MCO_K/\qf_1^{s_1}\MCO_K \times \cdots \times \MCO_K/\qf_t^{s_t}\MCO_K.
  \end{eqnarray} 

In fact, since $G$ acts trivially on each $\qf_i$, we have an induced action of $G$ on each $\MCO_K/\qf_i^{s_i}\MCO_K$, and the second isomorphism above is an isomorphism of $G$-modules.

The following observation is elementary but key to understanding the possible quotients $\Lambda/\Jc$.  
\begin{lemma} \label{fund_lemma}
With $\Lambda$ and $\Ic$ as above, we have that
\[
\Lambda/\Ic\Lambda \cong \mathcal{R}_1 \times \cdots \times \mathcal{R}_t
\]
where $\mathcal{R}_i$  is the ring $\oplus_{j=0}^{n-1} (\MCO_K/\qf_i^{s_i}\MCO_K)z^j$ is subject to $z(k+ \qf_i^{s_i}\MCO_K) = (\sigma(k)+\qf_i^{s_i}\MCO_K) z$ and $z^n = u + \qf_i^{s_i}$.
\end{lemma}

\begin{proof}
Using the isomorphism (\ref{CRTIsoms2}), we may describe $\Lambda/\Ic\Lambda$ as the ring
\[
\oplus_{i=0}^{n-1}(k_{i,1} + \qf_1^{s_1}\MCO_K,\dots,k_{i,t} + \qf_t^{s_t}\MCO_K) z^i.
\]  subject to the relations $z (k_1 + \qf_1^{s_1}\MCO_K,\dots,k_t + \qf_t^{s_t}\MCO_K) = (\sigma(k_1) + \qf_1^{s_1}\MCO_K,\dots, \sigma(k_t) + \qf_t^{s_t}\MCO_K) z$, and $z^n = (u + \qf_1^{s_1},\dots, u + \qf_t^{s_t})$.
We map the typical summand $(k_{j,1} + \qf_1^{s_1}\MCO_K,\dots,k_{j,t} + \qf_t^{s_t}\MCO_K)z^j$ to the element 
\[
\left( (k_{j,1} + \qf_1^{s_1}\MCO_K)z^j, \dots, (k_{j,t} + \qf_t^{s_t}\MCO_K)z^j \right).
\]
It is easy to see that this map is an injective ring homomorphism.  For surjectivity, we note that the ring $\mathcal{R}_1 \times \cdots \times \mathcal{R}_t$ is a sum of elements of the form $\left( 0, \dots, 0, (k_i + \qf_i^{s_i}\MCO_K) z^{l_i},0,\dots,0\right)$, and an element of this form is the image of   $\left(0,\dots, 0,k_i + \qf_i^{s_i}\MCO_K,0,\dots,0\right)z^{l_i}$. 
\end{proof}

\begin{definition} \label{gen_cyc_alg_defn}
Given an extension of commutative rings $R\subset S$, a finite cyclic group $G = \langle \sigma \rangle$ of order $n$ acting on $S$ such that the action is trivial on $R$, and given an element $u \in R$ (not necessarily invertible, possibly even zero), we refer to the ring $\oplus_{i=0}^{n-1} S z^i$ subject to the relations $zs = \sigma(s)z$ for all $s\in S$, and $z^n = u$, as a generalized cyclic algebra, and denote it $(S/R, \sigma, u)$, or $\left(\dfrac{S}{R}, \sigma, u\right)$ when needed.
\end{definition}

Note that $R$ is naturally contained in the center of $(S/R, \sigma, u)$.  (The center could be larger in general.)  Thus, $(S/R, \sigma, u)$ is naturally an $R$-algebra.

Note too that in our definition of generalized cyclic algebras above, we do not require that $G$ act{s} faithfully on $S$.  Thus, there may be a non-trivial subgroup of $G$ that acts trivially on $S$.  

The rings $\mathcal{R}_i$ in the statement of Lemma \ref{fund_lemma} are generalized cyclic algebras of the form $\left(\dfrac{\MCO_K/\qf^s\MCO_K }{ \MCO_F/\qf^s\MCO_F}, \sigma, u+\qf^s\right)$. Thus, 
$\Lambda/\Ic$ decomposes as the direct product of {generalized} cyclic algebras $({S/R}, \sigma, u)$, where the rings $R\subseteq S$ involved are all finite rings.  Under the isomorphism of the lemma, an ideal $I$ of $\Lambda/\Ic$ (itself arising as $\Jc/\Ic$ for some ideal $\Jc$ of $\Lambda$), will correspond to an ideal of the direct product 
$\mathcal{R}_1 \times \cdots \times \mathcal{R}_t$.
  Now, it is elementary that every ideal of  $\mathcal{R}_1 \times \cdots \times \mathcal{R}_t$ is of the form $I_1 \times I_2 \times \dots \times I_t$ where each $I_j$ is an ideal of $\mathcal{R}_j$. As a result, to identify an ideal $\Jc$ of $\Lambda$ that contains a fixed ideal $\Ic$ of $\MCO_F$ (and thus to identify the quotient $\Lambda/\Jc$), we may
just focus our attention on the case where $\Ic$ is a power of a single prime ideal, i.e., $\Ic = \qf^s$, where $\qf$ is some prime ideal of {$\MCO_F$},
and 
\[
\Lambda/\Ic\Lambda \cong \left(\dfrac{\MCO_K/\qf^s\MCO_K }{ \MCO_F/\qf^s\MCO_F}, \sigma, u+\qf^s \right) = \oplus_{j=0}^{n-1} (\MCO_K/\qf^s\MCO_K)z^j.
\]   

The ring $\left(\dfrac{\MCO_K/\qf^s\MCO_K }{ \MCO_F/\qf^s\MCO_F}, \sigma, u+\qf^s \right) $ is  a finite ring, of cardinality $n |\MCO_K/\qf^s\MCO_K|$.  In principle, therefore, the ideals of any one such ring should be easy to determine, by brute-force if necessary.  In many cases, however, the cyclic algebra $\left(\dfrac{\MCO_K/\qf^s\MCO_K }{ \MCO_F/\qf^s\MCO_F}, \sigma, u+\qf^s \right) $ reveals itself as isomorphic to one of a family of more familiar rings, whose ideals are already known.  We will explore such cases in the sections ahead.  

One phenomenon that qualitatively distinguishes one instance of a cyclic algebra of the form $\left(\dfrac{\MCO_K/\qf^s\MCO_K }{ \MCO_F/\qf^s\MCO_F}, \sigma, u+\qf^s \right) $ from another is how the prime ideal {$\qf$} extends to $K$.
Recall that since $K/F$ is a Galois extension of degree $n$, we have that $n=efg$, where $g$ counts the number of primes in the factorization of $\qf\MCO_K$, $e$ is the ramification index, and $f$ the inertial degree. We will focus in this paper on the \textit{unramified} situation, i.e., where $e=1$.

\begin{remark} \label{cyc_alg_skew_poly}
We wish to point out that the {generalized} cyclic algebra $(S/R, \sigma, u)$ can also be described as a quotient of a skew-polynomial ring with coefficients in $S$.  Recall that given a ring $S$ with a group $G = \langle \sigma \rangle$ acting on it, the skew polynomial ring $S[x; \sigma]$ is the set of all polynomials $s_0 + s_1 x + \cdots + s_t x^t$ ($s_i\in S$, $t\ge 0$), with multiplication twisted by the relation $x s = \sigma (x) s$ for all $s\in S$.  It is clear that $(S/R, \sigma, u)$ is the quotient of $S[x; \sigma]$ by the two-sided ideal generated by the polynomial $x^n-u$. Note that since $x^n$ acts trivially on $S$ in our situation, $x^n -u$ is actually in the center of $S[x; \sigma]$.
\end{remark}

\begin{remark}\label{rem:residue_matrices} 
In the case where $\Jc$ is the ideal generated an ideal $\Ic$ of $\Oc_F$, the quotient $\LA/\Ic\LA$ can naturally be thought of as a subring of a suitable matrix ring as follows: First note that  $\LA \subseteq \Mc_n(\Oc_K)$ via (\ref{eq:mat}).  Consider the ideal generated by $\Ic\subseteq \Oc_F\subseteq \Oc_K$ in $\Mc_n(\Oc_K)$: this is the ideal $\Mc_n(\Ic\Oc_K)$.  It is easy to see from (\ref{eqn:ILA_desc}) and the nature of the embedding defined in (\ref{eq:mat}) that $\Mc_n(\Ic\Oc_K) \cap \LA$ is precisely $\Ic\LA$.  It follows that $\LA/\Ic\LA \subseteq \Mc_n(\Oc_K)/ \Mc_n(\Ic\Oc_K) \cong \Mc_n(\Oc_K/\Ic\Oc_K)$.
Let $\bar{x}_{i}$ be the image of $x_i \in \Lambda$ in $\Lambda/\Ic\Lambda$, $i=1,\ldots,L$.   It follows  that we may view $\bar{x}_i$ as naturally embedded in $\Mc_n(\Oc_K/\Ic\Oc_K)$ by taking the representation $M(x_i)$ and modding the entries by $\Ic\Oc_K$.  We write this embedding as $M(\bar{x_i})$. 
Hence, the typical space-time codeword $(M(x_1),\ldots,M(x_L)) \in \Cc$ in our coset code construction is a preimage under $\pi: \oplus_{i=0}^{L} \LA \rightarrow \oplus_{i=0}^{L} \LA/\Jc \subseteq \oplus_{i=0}^{L}\Mc_n(\Oc_K/\Ic\Oc_K) $ of the codeword 
\begin{equation}\label{eq:MnK}
(M(\bar{x}_1),\ldots,M(\bar{x}_L)) \in \oplus_{i=1}^{L} \Mc_n(\bar{K}).
\end{equation}
We will have occasion to use this matrix representation later.
\end{remark}

\begin{remark}
If $I$ is an ideal of $\MCO_K/\qf^s\MCO_K $ that is sent to itself under $\sigma$, then the set $\oplus_{j=0}^{n-1} Iz^j$ will be an ideal of $\left(\dfrac{\MCO_K/\qf^s\MCO_K }{ \MCO_F/\qf^s\MCO_F}, \sigma, u+\qf^s \right) $. 
\end{remark}

%
%
%

\section{The Inertial Case: $\Ic=\qf$, $g=e=1$, $f=n$}
\label{sec:inert}

When $g=e=1$, $\qf$ remains prime in $\MCO_K$, that is $\qf\Oc_K=\Qf$ with inertial degree $f=n$.  Then $\MCO_K/\qf\Oc_K=\Oc_K/\Qf$ and $\Oc_F/\qf$ are finite fields, that we denote respectively $\overline{K}$ and $\overline{F}$, and in fact $\overline{K}/\overline{F}$ is a cyclic Galois extension of degree $n$ with generator $\sigmabar$, the induced action of $\sigma$. (All this is standard, and can be deduced from \cite[Chapter I, Theorem 6.6, or Chapter III, Theorem 1.4]{Jan} for instance.)
Thus
\[
\Lambda/\Ic\Lambda \cong \oplus_{j=0}^{n-1} (\MCO_K/\qf\MCO_K)z^j \cong \oplus_{j=0}^{n-1}\bar{K}z^j,
\]   
with $z (k + \Qf) = (\sigma(k)+\Qf)z$ and $z^n=\bar{u}=u+\qf$. 

\subsection{The case $u\not\in \qf$}

\begin{proposition} \label{prop:f=n} 
Suppose that $\Ic=\qf$ is a prime ideal in $\Oc_F$ which remains inert in $\Oc_K$: $\qf\Oc_K=\Qf$ for $\Qf$ a prime of $\Oc_K$. 
If $u\not\in \qf$, then $\Lambda/\Ic\Lambda \cong (\bar{K}/\bar{F}, \sigmabar, \bar{u}) \cong \Mc_n(\overline{F})$.  Moreover, the only proper two-sided ideal $\Jc$ of $\Lambda$  that contains $\Ic\ (=\qf)$ is $\Ic\Lambda = \oplus_{j=0}^{n-1} \qf\Oc_K z^j$.
\end{proposition} 
\begin{proof}
If $u\not\in \qf$, then $\ubar$ is a nonzero element of $\overline{F}$, and thus $\Lambda/\Ic\Lambda$ is the usual cyclic algebra $(\overline{K}/\overline{F}, \sigmabar, \ubar)$ of degree $n$,  
{which is}, by Wedderburn's theorem, $\overline{F}$-isomorphic to $M_t(D)$ for {a} suitable $t$, where $D$ is a division algebra with center $\overline{F}$.  By another theorem of Wedderburn, since $\overline{F}$ is finite, the only $\overline{F}$-central division algebra is $\overline{F}$ itself. (See \cite[Theorem 7.4, Theorem 7.24]{Reiner}, for instance, for Wedderburn's theorems.) We therefore find that $\Lambda/\Ic\Lambda \cong \Mc_n(\overline{F})$.  Next, note that $\Mc_n(\overline{F})$ is a simple ring, i.e., it has no two-sided ideals other than the zero ideal and the whole ring. It follows from the correspondence between ideals of $\Lambda$ that contain $\Ic\Lambda$ and ideals of $\Lambda/\Ic\Lambda$ that the only proper two-sided ideal $\Jc$ of $\Lambda$ that contains $\Ic$  is $\Ic\Lambda = \oplus_{j=0}^{n-1} \qf\Oc_K z^j$.

\end{proof}

{The cyclic algebras considered in the three examples below are division, see \cite{perfect} for a proof.
\begin{example}\label{ex:q5} 
Consider the cyclic division algebra $\Dc=(\QQ(\imath,\sqrt{5})/\QQ(\imath),\sigma,\imath)$.  (Thus,  $K =\QQ(\imath,\sqrt{5}) $ and $F=\QQ(\imath)$ here.) The ring of integers in $\QQ(\imath)$ and $\QQ(\sqrt{5})$ are $\ZZ[\imath]$ and $\ZZ[(1+\sqrt{5})/2]$, respectively(\cite[Chapter I, Theorem 9.2]{Jan}).  It follows from \cite[Chapter I, Corollary 9.4]{Jan} for instance that the ring of integers in $K$ is $\ZZ[\imath,(1+\sqrt{5})/2]$.  We have
\[
\Lambda = \ZZ[\imath,(1+\sqrt{5})/2]\oplus\ZZ[\imath,(1+\sqrt{5})/2]z,~z^2=\imath.
\]
 Take $\qf = \langle 1+\imath \rangle$ in $\Oc_{\QQ(\imath)} = \ZZ[\imath]$.
Then $\Ic=\langle 1+\imath \rangle$ and $\ZZ[\imath]/\Ic\cong \FF_2$. The ideal $\Ic$ is inert in $\QQ(\imath,\sqrt{5})$ (the element $(1+\sqrt{5})/2$ modulo $(1+\imath)\Oc_{\QQ(\imath,\sqrt{5})}$ generates the field extension $\FF_4\supseteq\FF_2$, since it satisfies $x^2 - x - 1$), and we have by Proposition \ref{prop:f=n} above
\[
\Lambda/\Ic\Lambda\cong \Mc_2(\FF_2).
\]
Thus, by the last statement of the proposition above, $\Jc = \Ic\Lambda=(1+\imath) \Lambda$ is the only ideal of $\Lambda$ which contains $\langle 1+\imath \rangle\LA$, and
\[
\Lambda/\Jc \cong \Mc_2(\FF_2),
\]
as also shown in \cite{Luzzi}.
\end{example}
\begin{example}\label{ex:q7} 
Let $\omega_7$ be a primitive 7th root of unity. 
Consider the cyclic division algebra $\Dc=(\QQ(\omega_3,\omega_7+\omega_7^{-1})/\QQ(\omega_3),\sigma,\omega_3, \sigma, \omega_3)$. (Thus,  $K =\QQ(\omega_3,\omega_7+\omega_7^{-1}) $ and $F=\QQ(\omega_3)$ here.)
The ring of integers of $\QQ(\omega_3)$ is $\ZZ[\omega_3]$ (\cite[Chapter I, Theorem 9.2]{Jan}). The ring of integers of the field $\QQ(\omega_7+\omega_7^{-1})$ is $\ZZ[\omega_7+\omega_7^{-1}]$ (see \cite[Proposition 2.16]{Wash} for instance).  By \cite[Chapter I, Corollary 9.4]{Jan} for instance, the ring of integers of $K$ is $\ZZ[\omega_3,\omega_7+\omega_7^{-1}]$.
We have
\[
\Lambda = \oplus_{j=0}^{2}\ZZ[\omega_3,\omega_7+\omega_7^{-1}]z^j,~z^3=\omega_3.
\]

Consider the prime ideal $\langle 2 \rangle$ of $\ZZ$. This ideal stays prime in $\Oc_{\QQ(\omega_3)} =\ZZ[\omega_3]$ (note that $\omega_3$ satisfies $x^2 + x + 1$ over $\QQ$, and hence generates the field extension $\FF_4\supseteq\FF_2$ modulo $2\ZZ[\omega_3]$).  We now 
take $\Ic=\qf = \langle 2 \rangle$ in $\Oc_{\QQ(\omega_3)}$. The ideal $\Ic$ remains prime in $\Oc_{\QQ(\omega_3,\omega_7+\omega_7^{-1})} = \ZZ[\omega_3,\omega_7+\omega_7^{-1}] $.  (This can be seen by noting that
$\omega_7+\omega_7^{-1}$ satisfies $x^3+x^2-2x-1$ over $\QQ$ and hence generates the field extension $\FF_8\supseteq\FF_2$ modulo $2 \ZZ[\omega_7+\omega_7^{-1}]$.  Thus, the residue of  $\ZZ[\omega_3,\omega_7+\omega_7^{-1}] $ contains the compositum  of $\FF_4 = \FF_{2^2}$ and $\FF_8 = \FF_{2^3}$, i.e., $\FF_{2^6}=\FF_{64}$.)  Hence, by Proposition \ref{prop:f=n} above
\[
\Lambda/\Ic\Lambda\cong \Mc_3(\FF_4).
\]

Thus, by the last statement of the proposition above, $\Jc = \Ic\Lambda=2 \Lambda$ is the only ideal of $\Lambda$ which contains $2\LA$, and
\[
\Lambda/\Jc\Lambda\cong \Mc_3(\FF_4),
\]
as also shown in \cite{pstm}.
\end{example}
\begin{example}\label{ex:q15} 
Let $\omega_{15}$ be a primitive 15th root of unity.
Consider the cyclic division algebra $\Dc=(\QQ(\imath,\omega_{15}+\omega_{15}^{-1})/\QQ(\imath),\sigma,\imath)$.  (Thus,  $K = \QQ(\imath,\omega_{15}+\omega_{15}^{-1})$ and $F= \QQ(\imath)$ here.)  As in the previous example, we find that the ring of integers of $K$ is $\ZZ[\imath,\omega_{15}+\omega_{15}^{-1}]$. We take
\[
\Lambda = \oplus_{j=0}^{3}\ZZ[\imath,\omega_{15}+\omega_{15}^{-1}]z^j,~z^4=\imath.
\]

Consider the ideal $\langle 2 \rangle $ of $\ZZ$.  It becomes the square of the ideal $\langle 1+\imath \rangle$ in $\Oc_{\QQ(\imath)} = \ZZ[\imath]$ (thus, $\langle 2 \rangle$ is totally ramified in $\QQ(\imath)$ and $\ZZ[\imath]/\langle 1 + \imath \rangle \cong \FF_2$).  Now take $\Ic = \qf = \langle 1+\imath \rangle$ in $\ZZ[\imath]$.  The prime $\langle 1+\imath \rangle$ remains prime in $\Oc_{\QQ(\imath,\omega_{15}+\omega_{15}^{-1})} = \ZZ[\imath,\omega_{15}+\omega_{15}^{-1}]$.  
(This can be seen by first noting that $\langle 2 \rangle$ remains prime in the extension $\QQ(\omega_{15}+\omega_{15}^{-1})/\QQ$. This follows from the fact that $\omega_{15}+\omega_{15}^{-1}$ satisfies $x^4 - x^3 - 4x^2 + 4x + 1$ over $\QQ$, and hence generates the field extension $\FF_{16} \supseteq \FF_2$. Next, by considering the fact that $\QQ(\imath,\omega_{15}+\omega_{15}^{-1})$ is the compositum of the totally ramified extension $\QQ(\imath)/\QQ$ and the inert extension $\QQ(\omega_{15}+\omega_{15}^{-1})/\QQ$, we find that $\langle 1+\imath \rangle$ stays prime in $\ZZ[\imath,\omega_{15}+\omega_{15}^{-1}]$.)

Hence, by Proposition \ref{prop:f=n} above
\[
\Lambda/\Ic\Lambda\cong \Mc_4(\FF_2).
\]

Thus, by the last statement of the proposition above, $\Jc = \Ic\Lambda=(1+\imath) \Lambda$ is the only ideal of $\Lambda$ which contains $(1+\imath)\LA$ in $\ZZ[\imath]$, and
\[
\Lambda/\Jc\Lambda\cong \Mc_4(\FF_2),
\]
as also shown in \cite{pstm}.
\end{example}

These three examples cover the cases studied in the context of space-time coded modulation \cite{Luzzi,pstm}.


Recall that the design criterion for space-time coded modulation is the minimum determinant, which from (\ref{eq:dhbound}) satisfies
\[
\Delta_{min} \geq \min\left( d_H(\bar{\Cc})^2 \min_{x_i \not\in\Jc} |\det(X_i)|^2, \min_{0\neq x_i \in \Jc} |\det(X_i)|^2\right),
\]
where $X_i=M(x_i)$.
In the three examples above $\Jc=(\alpha)$ is a principal ideal with 
$\alpha\in \Oc_F$, so $M(x_i)=M(\alpha) M(x'_i)$ for some $x'_i\in \LA$. Hence 
\begin{equation}\label{eq:minprin}
\Delta_{min} \geq \min_{0\neq X_i} |\det(X_i)|^2 \min\left( d_H(\bar{\Cc})^2 , |\alpha|^{2n}\right),
\end{equation}
since 
\[
M(\alpha)=
\begin{bmatrix}
\alpha & & \\
       & \ddots & \\
       &        & \alpha 
\end{bmatrix},~\text{and therefore} \ \det(M(\alpha))=\alpha^n.
\]
Thus, to increase the lower bound for $\Delta_{min}$ we need to construct examples where $d_H(\bar{\Cc})$ is at least $|\alpha|^n$.  So, for Example \ref{ex:q5} we want a code $\bar{\Cc}$ over 
$\Mc_2(\FF_2)$ with $d_H(\bar{\Cc})$ at least $2$, for Example \ref{ex:q7} a code $\bar{\Cc}$ over $\Mc_3(\FF_4)$ with $d_H(\bar{\Cc})$ at least $2^3$ and for Example \ref{ex:q15} a code $\bar{\Cc}$ over $\Mc_4(\FF_2)$ with $d_H(\bar{\Cc})$ at least $2^2$.

\begin{example}
We continue Example \ref{ex:q5}, and want to construct a space-time codeword $X=(X_1,X_2,X_3)$. By Remark \ref{rem:residue_matrices}, a codeword from $\bar{\Cc}$ is of the form
\[
( M(\bar{x}_1),~M(\bar{x}_2),~M(\bar{x}_3)),
\]
where
\[
M(\bar{x}_i)=
\begin{bmatrix}
\bar{x}_{i,0} & \bar{u}\bar{\sigma}(\bar{x}_{i,1}) \\
\bar{x}_{i,1} & \bar{\sigma}( \bar{x}_{i,0}) \\
\end{bmatrix},~i=1,2,3.
\]
Since we only need a Hamming distance $d_H(\bar{\Cc})=2$, a simple parity code is enough. We choose $\bar{\Cc}$ to be the subspace of $\oplus_{i=1}^{3} \LA/\Jc \subseteq \oplus_{i=1}^{3} \Mc_2(\FF_4)$ given by $M(\bar{x}_3) = M(\bar{x}_1) + M(\bar{x}_2)$.  We then choose $\Cc$ to be the submodule of triples $(x_1, x_2, x_3) \in \LA$ satisfying $x_3 = x_1 + x_2$.
This gives the following space-time codeword:
\[
( M(x_1),~M(x_2),~M(x_1)+M(x_2))=(M(x_1),M(x_2),M(x_3)),~x_3=x_1+x_2,
\]
namely
\[
\left(
\begin{bmatrix}
x_{1,0} & \imath\sigma(x_{1,1}) \\
x_{1,1} & \sigma( x_{1,0}) \\
\end{bmatrix},~
\begin{bmatrix}
x_{2,0} & \imath\sigma(x_{2,1}) \\
x_{2,1} & \sigma( x_{2,0}) \\
\end{bmatrix},~
\begin{bmatrix}
x_{1,0}+x_{2,0} & \imath\sigma(x_{2,1}+x_{1,1}) \\
x_{1,1}+x_{2,1} & \sigma( x_{1,0}+x_{2,0}) \\
\end{bmatrix}
\right).
\]
We observe here that the lower bound 
\[
\Delta_{min}=\min_{0\neq X}|\det(\sum_{i=1}^3 X_iX_i^*)| \geq 4 \min_{x_i\neq 0}|\det(M(x_i)) |^2
\]
from (\ref{eq:minprin}) is actually achieved  with the codeword $(M(x_1),0,M(x_1))$, when $x_1$ is chosen such that $|\det(M(x_1))|^2$ is minimized.
%
\end{example} 


\begin{example} \label{ex:general_scheme}  
We describe here a general scheme for constructing a code $\Cc$ such that $\bar{\Cc}$ has a prescribed minimum distance, in the special case where $\Jc = \Ic\LA = \qf\LA$, and the extension $K/F$ is inertial with respect to $\qf$.  This scheme does not take full advantage of the structure of the residue ring $\LA/\Ic\LA$, but all the same, has the advantage of flexibility.  It works both when $u \not\in\qf$ (as in this subsection), as well as the case when $u\in\qf$ (\S \ref{subsectn:inert_ubar_zero} ahead).  It can also be adapted to the case when $\qf$ splits (with no ramification, as in \S \ref{sec:noram} ahead), as well as to other cases.

Our starting point is  (\ref{eq:direct_sum}), {namely}:
\[
\LA/\Ic\LA \cong \oplus_{i=0}^{n-1} (\MCO_K/\Ic\MCO_K)z^i.
\]

In our situation, $\Oc_K/\Ic\Oc_K = \bar{K}$ is a field.  We choose a code $\Bc$ of the desired minimum distance $d_H$ and desired length $L$ over $\bar{K}$.  Thus, $\Bc$ is a subgroup of $\oplus_{i=1}^L \bar{K}$. Our strategy will be to incorporate the entries of the code $\Bc$ into just the first summand of $\LA/\Ic\LA$ on the right side of (\ref{eq:direct_sum}). Recall from Remark \ref{rem:residue_matrices} above that   we may view $\LA/\Ic\LA$ as a subset of $\Mc_n(\bar{K})$, obtained by modding out the entries of the matrix in (\ref{eq:mat}) by $\Ic\Oc_K$. As in that remark, we write $M(\bar{x}_i)$ for the matrix obtained by modding out the matrix $M(x_i)$ corresponding to $x_i\in\LA$, and we write $\bar{x}_{i,j}$ for the entries in the first column of $M(\bar{x}_i)$.    We will choose as our outer code $\bar{\Cc} \subseteq \oplus_{i=1}^L \LA/\Ic\LA \subseteq \oplus_{i=1}^L \Mc_n(\bar{K})$ those $L$-tuples of matrices coming from $\LA/\Ic\LA$ such that $(\bar{x}_{1,0}, \dots, \bar{x}_{L,0})$ belong to the code $\Bc$.  (Thus, the remaining entries $\bar{x}_{i,j}$ of the first columns of $M(\bar{x}_i)$ are allowed to be arbitrary for $j=1, \dots, n-1$, and the remaining columns then follow the pattern of the matrix in (\ref{eq:mat}).) We will now define $\Cc$ to be $\pi^{-1} (\bar{\Cc})$, where we recall that $\pi$ is the natural projection $\oplus_{i=1}^L \LA \rightarrow \oplus_{i=1}^L \LA/\Ic\LA$.  It is clear that $\bar{\Cc}$ has the desired minimum distance $d_H$, since the $L$-tuple of $(1,1)$ entries of the matrices $M(\bar{x}_i)$ comes from the code $\Bc$.

\end{example}

\begin{example} We illustrate Example \ref{ex:general_scheme} above in the context of
Example \ref{ex:q15}. We want to construct a space-time codeword $X=(X_1,\ldots,X_L)$ over the ring $\LA$ of that example such that the code $\bar{\Cc}$ over $\LA/\Ic\LA$ has Hamming distance 4. To do so, we  may use an $(L,k)$ Reed-Solomon code over $\FF_{16}$, that is a subspace of dimension $k$ of $\oplus_{i=1}^L\FF_{16}$. The minimum distance of a Reed-Solomon code is $L-k+1$, we thus choose $k$ and $L$ such that $L-k+1=4$ and $L \leq 16$. To obtain a codeword of length $L$, take a polynomial $p(x)\in\FF_{16}[x]$ of degree $k-1$, and evaluate it in $L$ distinct values of $\FF_{16}$. We may take $\Bc$ of the previous example to be this Reed-Solomon code to obtain the desired code $\Cc$ over $\LA$.
\end{example}

A similar construction can be given as a concrete instance of Example \ref{ex:q7}.

\subsection{The case $u\in \qf$} \label{subsectn:inert_ubar_zero}

\begin{proposition} \label{prop:inert_s_one_ubar_zero} 
Suppose that $\Ic=\qf$ is a prime ideal in $\Oc_F$ which remains inert in $\Oc_K$: $\qf\Oc_K=\Qf$ for $\Qf$ a prime of $\Oc_K$. 
If $u\in \qf$, then 
\[
\Lambda/\Ic\Lambda \cong (\bar{K}/\bar{F},\sigmabar, 0) \cong \overline{K}[x,\sigmabar]/\langle x^n\rangle,
\]
where $\overline{K}[x,\sigmabar]$ is the ring of twisted polynomials with indeterminate $x$ and coefficients in $\overline{K}$. Furthermore, the only possible proper ideals $\Jc$ of $\LA$ that contain $\Ic  = \qf$ are the ideals $\Jc_i = \langle z^i \rangle $, $i=1, \dots, n-1$, and for each such $\Jc_i$,
\[
\LA/\Jc_i \cong \overline{K}[x,\sigmabar]/\langle x^i\rangle \cong\oplus_{j=0}^{i-1} \overline{K} z^j,
\]
subject to $z \overline{k} = \sigmabar(\overline{k}) z$ and $z^i = 0$. In particular, $\LA/\Jc_1\cong \overline{K}$. 
\end{proposition}

\begin{proof}
If $u\in \qf$, then $\ubar = 0$, so $\Lambda/\Ic\Lambda$ is the ring $\oplus_{j=0}^{n-1} \overline{K} z^j$, subject to $z \overline{k} = \sigmabar(\overline{k}) z$ and $z^n = 0$. This is just the generalized cyclic algebra $(\bar{K}/\bar{F},\sigmabar, 0)$.  It is also the quotient of the twisted polynomial ring $\overline{K}[x,\sigmabar]$ (where $x \overline{k} = \sigmabar(\overline{k}) x$, $\bar{k}=k+\Qf$) by the ideal $(x^n)$ (see Remark \ref{cyc_alg_skew_poly}).  Note that the coefficient ring $\overline{K}$ is a field, and that since $\overline{K}/\overline{F}$ is a Galois extension of degree $n$, $\langle \sigma \rangle$ acts faithfully on $\overline{K}$, that is, $\sigma^i = \text{id}$ implies $i$ is a multiple of $n$.

The two-sided ideals of the twisted polynomial ring $\overline{K}[x,\sigmabar]$ are well known \cite[Prop. 1.6.25]{Rowen}: they are of the form $\langle p(x^n)x^i \rangle$, where $p(x^n) \in\overline{F}[x^n]$, which is the center of $\overline{K}[x,\sigmabar]$.  Thus, the proper two-sided ideals of $\overline{K}[x,\sigmabar]$ that contain $\langle x^n \rangle$ are only those of the form $\langle x^i \rangle$, $i=1,\dots, n$, to which correspond ideals $\Jc_i$ of $\overline{K}[x,\sigmabar]/\langle x^n\rangle$ of the form $\langle x^i\rangle/\langle x^n\rangle$ (which in turn correspond to ideals $\langle z^i\rangle+\Ic\Lambda$ of $\Lambda/\Ic\Lambda$), so that
\[
\LA/\Jc_i \cong \overline{K}[x,\sigmabar]/\langle x^i\rangle
\]
is a ring of the form $\oplus_{j=0}^{i-1} \overline{K} z^j$, subject to $z \overline{k} = \sigmabar(\overline{k}) z$ and $z^i = 0$. 
Exactly as in the proof of Proposition \ref{prop:f=n}, we find from the correspondence between ideals of $\Lambda$ that contain $\Ic\Lambda$ and the ideals of $\Lambda/\Ic\Lambda$ that these ideals $\Jc_i$ are the only two-sided ideals of $\Lambda$ that contain $\Ic = \qf$.

Notice that taking $i=1$ we find $\overline{K}$ to be a possible quotient of $\LA$.  
\end{proof}

Suppose that $\theta$ is an algebraic integer such that $\QQ(\theta,\imath)/\QQ(\imath)$ is cyclic of degree $n$, $\Dc=(\QQ(\theta,\imath)/\QQ(\imath),\sigma,1+\imath)$ is division, and $\langle 1+\imath \rangle$ is inert in the extension $\QQ(\theta,\imath)/\QQ(\imath)$. We pick $\Ic=\langle 1+\imath\rangle $.  Then $u=1+\imath \in \Ic$, and the above proposition tells us that the possibilities for  $\Jc$ are the ideals $<z^j>$, $j=1,\ldots,n$. Extreme cases are $\Jc_n=\langle z^n \rangle= (1+\imath)\Lambda$ in which case 
\[
\Lambda/\Jc_n \cong \bar{K}\oplus \ldots \oplus \bar{K}z^{n-1}
\] subject to $\bar{z}^n = 0$,
and $\Jc_1=\langle z \rangle = z\Lambda$, which interestingly yields 
\[
\Lambda/\Jc_1 \cong \bar{K}.
\]
In the latter case, a space-time codeword $X=(M(x_1),\ldots,M(x_L))$ is sent to the codeword 
\[
(\bar{x}_{1,0},\bar{x}_{2,0},\ldots,\bar{x}_{L,0})\in \oplus_{i=1}^L\bar{K}
\] 
under the projection $\pi: \oplus_{i=1}^L \LA \rightarrow \oplus_{i=1}^L \LA/\Jc_1 \cong \oplus_{i=1}^L\bar{K}$,
where $x_i=\sum_{j=0}^{n-1}x_{i,j}z^j$.
\begin{example}\label{ex:q52}
Consider the same field extension $\QQ(\imath,\sqrt{5})/\QQ(\imath)$ as in Example \ref{ex:q5} and the cyclic division algebra $\Dc=(\QQ(\imath,\sqrt{5})/\QQ(\imath),\sigma,1+\imath)$, with 
\[
\Lambda = \ZZ[\imath,(1+\sqrt{5})/2]\oplus\ZZ[\imath,(1+\sqrt{5})/2]z,~z^2=\imath+1.
\]
As seen in Example \ref{ex:q5}, the ideal $\Ic=\qf=\langle 1+\imath \rangle$ is inert in the extension $\QQ(\imath,\sqrt{5})/\QQ(\imath)$.  The $2$-adic valuation on $\QQ$ extends to the $(1+\imath)$-adic valuation on $\QQ(\imath)$ (as $2$ is totally ramified in the extension $\QQ(\imath)/\QQ$).  We may take the value group of  $\QQ(\imath)$ to be $\ZZ$ and the value of $(1+\imath)$ to be $1$ in this valuation.  The  $(1+\imath)$-adic valuation then extends uniquely to $\QQ(\imath,\sqrt{5})$, and the value group of $\QQ(\imath,\sqrt{5})$ is also $\ZZ$.  It follows that if the norm of $x \in \QQ(\imath,\sqrt{5})$ equals $(1+\imath)$, then $v(x \sigma(x))= v(x) + v(\sigma(x)) = 2v(x) = v(1+\imath) = 1$.  But this is an impossible equation in $\ZZ$.  Thus, $\Dc$ is indeed a division algebra.

We are in the situation covered by Proposition \ref{prop:inert_s_one_ubar_zero}.  We consider {$\Jc_1=\langle z \rangle = z\Lambda$. }

Using the above proposition, we have
\[
\Lambda/\Jc_1\Lambda \cong \FF_4.
\]
\end{example}

To further evaluate the bound (\ref{eq:dhbound}) given by 
\[
\Delta_{min} \geq \min\left( d_H(\bar{\Cc})^2 \min_{ x_i \not \in \Jc} |\det(X_i)|^2, \min_{0\neq x_i \in \Jc} |\det(X_i)|^2\right),
\]
where we recall that $X_i=M(x_i)$, note that for $x_i \in \Jc_j=z^j\Lambda$, $M(x_i)=M(z^j x_i')$, for some $x_i' \in{\LA}$. 
Then 
\[
M(z^j)=
\begin{bmatrix}
0 &    &          &  u   \\
1 & 0  &          &     \\
   &    &  \ddots  &     \\
   &    &         1 & 0  \\  
\end{bmatrix}^j ,~\det(M(z^j))=u^j,
\]
so 
\begin{equation}
\Delta_{min} \geq \min_{0\neq x_i} |\det(X_i)|^2 \min\left( d_H(\bar{\Cc})^2, |u|^{2j} \right).
\end{equation}
In particular, taking $j=1$, we find for the codes constructed from $\Jc_1$ (where $\Lambda/\Jc_1\Lambda \cong \FF_4$),
\begin{equation}\label{eq:detminu}
\Delta_{min} \geq \min_{0\neq x_i} |\det(X_i)|^2 \min\left( d_H(\bar{\Cc})^2, |u|^{2} \right).
\end{equation}

\begin{example}
We continue with Example \ref{ex:q52}, and construct a space-time codeword $X=(X_1,X_2,X_3)$. Since $|u|^2=|1+\imath|^2=2$, it is enough to have $d_H(\bar{\Cc})=2$, so a single parity code can be used again, this time over $\FF_4$, and the codeword
\[
(\bar{x}_{1,0},\bar{x}_{2,0},\bar{x}_{1,0}+\bar{x}_{2,0}) \in \FF_4^3
\]
can be lifted to
\[
\left(
\begin{bmatrix}
x_{1,0} & (1+\imath)\sigma(x_{1,1}) \\ 
x_{1,1} & \sigma(x_{1,0})
\end{bmatrix},~
\begin{bmatrix}
x_{2,0} & (1+\imath)\sigma(x_{2,1}) \\ 
x_{2,1} & \sigma(x_{2,0})
\end{bmatrix},~
\begin{bmatrix}
x_{1,0}+x_{2,0} & (1+\imath)\sigma(x_{3,1}) \\ 
x_{3,1} & \sigma(x_{1,0}+x_{2,0})
\end{bmatrix}\right).
\]
From (\ref{eq:detminu}), a lower bound on $\Delta_{min}$ is 
\[
\Delta_{min}\geq  2\min_{0\neq x_i}|\det(X_i)|^2.
\]
This lower bound is actually achieved by  the codeword 
\[
\left(
\begin{bmatrix}
0 & 0 \\
0 & 0 \\
\end{bmatrix},~
\begin{bmatrix}
0 & 0 \\
0 & 0 \\
\end{bmatrix},~
\begin{bmatrix}
0 & (1+\imath)) \\ 1 & 0
\end{bmatrix}\right).\]

Indeed, for any $X = M(x)$, where $x\in \LA$, the minimum of $|\det(X)|$ is bounded below by $1$ since this determinant lies in $\ZZ[
\imath]$. Morever, the element $x=1$ of $\LA$ actually achieves this minimum.  It is clear that the codeword above has determinant $2$, the least possible.
\end{example}

%
%
%
\section{The Inertial Case: $\Ic = \qf^s$ ($s > 1$), $g=1$, $e=1$, $f=n$}
\label{sec:inert_power_of_prime}

We consider now a modification of the situation considered in Section \ref{sec:inert}: we assume that the ideal $\Ic$ of $\Oc_F$ is a \textit{power} of a prime ideal.  Thus, we assume that $\Ic = \qf^s$ for some prime ideal $\qf$ of $\Oc_F$ and some integer $s > 1$.  We continue to assume that $\qf$ stays prime in $\Oc_K$, that is, $g=1$, $e=1$, and $f=n$.  As in the previous section, we write $\Qf$ for the prime ideal $\qf \Oc_K$.

The ring $\Lambda/\Ic\Lambda$ is described by
\[
\Lambda/\Ic\Lambda \cong \oplus_{j=0}^{n-1} (\MCO_K/\qf^s\MCO_K)z^j \cong \oplus_{j=0}^{n-1} (\MCO_K/\Qf^s)z^j,
\]   
with $z (k + \Qf^s) = (\sigma(k)+\Qf^s)z$ and $z^n=u+\qf^s$.

We will study in this section the case where $u\not\in \qf$. 
We shall prove the following, which is a generalization of Proposition \ref{prop:f=n}:

\begin{proposition} \label{prop:inert_power_of_prime_u_not_in_q} In the situation described above ($\Ic = \qf^s$, $s>1$, $\qf$ remains prime in $\Oc_K$, $u\not\in\qf$), 
$\Lambda/\Ic\Lambda \cong \Mc_n(\Oc_F/\qf^s)$.   The two-sided proper ideals $\Jc$ of $\Lambda$ such that $\Jc$ contains $\qf^s\Lambda$ are the ideals $\Jc_t = \qf^t\Lambda = \oplus_{j=0}^{n-1} \qf^t\Oc_K z^j$, for $1 \le t \le s$, and their quotients satisfy $\Lambda/\Jc_t \cong \Mc_n(\Oc_F/\qf^t)$ .  
\end{proposition}

\begin{proof} Write $\widehat{F}$ and $\widehat{K}$ for the completions of $F$ and $K$ (respectively) at the primes $\qf$ and $\Qf$.  Then $\widehat{K}/\widehat{F}$ is cyclic, with Galois group (abusing notation) $\langle \sigma \rangle$. (Here $\widehat{K} \cong K \otimes_F \widehat{F}$, and (abusing notation) $\sigma$ acts on $\widehat{K}$ by acting trivially on $\widehat{F}$ and as $\sigma$ on $K$. All this is standard, and can be found, for instance, in the discussions in \cite[Chapter II, \S 3, \S 5, or Chapter III, \S 1]{Jan}.)  Write $\Vc_{\widehat{K}}$ and $\Vc_{\widehat{F}}$ respectively for the valuation rings of ${\widehat{K}}$ and ${\widehat{F}}$ under the $\qf$-adic valuation. 
The assumption $u\not \in \qf$ translates to the $\qf$-adic value of $u$ being zero, which then makes it a unit in $\Vc_{\widehat{F}}$.
The ring $\widehat{\Lambda} = \oplus_{j=0}^{n-1}\Vc_{\widehat{K}} z^j$, with relations $z x = \sigma(x) z$ for $x\in \Vc_{\widehat{K}}$, $z^n = u$ is then an Azumaya algebra over $\Vc_{\widehat{F}}$ (\cite[Example 2.4 (i)]{JW}). But Azumaya algebras over $\Vc_{\widehat{F}}$ are trivial (this follows from the fact that the Brauer group of $\bar{F}$ is zero since it is a finite field, and from the isomorphisms described in \cite[Theorem 2.8]{JW}). $\widehat{\LA}$ is hence isomorphic to the endomorphism ring of a finitely generated projective module over $\Vc_{\widehat{F}}$.  But $\Vc_{\widehat{F}}$ is a local ring, so projective modules over it are free, so $\widehat{\Lambda} \cong \Mc_n(\Vc_{\widehat{F}})$. It follows immediately that $\widehat{\Lambda}/\qf^s\widehat{\Lambda} \cong \Mc_n(\Vc_{\widehat{F}}/\qf^s)$.

Now write $\Vc_K$   for the valuation ring of $K$ with respect to the $\Qf$-adic valuation, and $\Vc_F$ for the valuation ring of $F$ with respect to the $\qf$-adic valuation.  (Thus, $\Vc_K$ is the localization of $\Oc_K$ at the prime ideal $\Qf$ and $\Vc_F$ is the localization of $\Oc_F$ at the prime ideal $\qf$.)  
We have  the relations $\Oc_K/\Qf^s \cong \Vc_K/\Qf^s \cong \Vc_{\widehat{K}}/\Qf^s$ and $\Oc_F/\qf^s\cong  \Vc_F/\qf^s \cong \Vc_{\widehat{F}}/\qf^s$ (see \cite[Chapter I, Lemma 3.1]{Jan} and the proof of  \cite[Chapter II, Theorem 3.8]{Jan} for instance).   Since $\widehat{\Lambda}/\qf^s\widehat{\Lambda} \cong \oplus_{j=0}^{n-1}(\Vc_{\widehat{K}}/\Qf^s) z^j$, with relations $z (x +\Qf^s) = (\sigma(x) + \Qf^s) z$ for $x\in \Vc_{\widehat{K}}/\Qf^s$ and $z^n = u+\qf^s$, we find $\widehat{\Lambda}/\qf^s\widehat{\Lambda} \cong \oplus_{j=0}^{n-1}(\Oc_{{K}}/\Qf^s) z^j$, with relations $z (x +\Qf^s) = (\sigma(x) + \Qf^s) z$ for $x\in \Oc_{K}/\Qf^s$ and $z^n = u+\qf^s$.  But this ring is just $\Lambda/\qf^s\Lambda$.  Thus, we find $\Lambda/\qf^s\Lambda \cong \Mc_n(\Vc_{\widehat{F}}/\qf^s) \cong \Mc_n(\Oc_{{F}}/\qf^s)$.

For the statements about the ideals $\Jc$, note that the ideals of $\Mc_n(\Oc_{{F}}/\qf^s)$ are the sets $\Mc_n(I)=I \cdot \Mc_n(\Oc_{{F}}/\qf^s)$, where $I$ is an ideal of $\Oc_{{F}}/\qf^s$. The proper ideals of $\Oc_F/\qf^s$ are the ideals $\qf^t/\qf^s$, $1 \le t \le s$. It follows that the proper ideals of $\Lambda/\qf^s\LA \cong \Mc_n(\Oc_{{F}}/\qf^s)$ are $(\qf^t/\qf^s) \Lambda/\qf^s\LA = \qf^t\LA /\qf^s\LA$. Since the proper ideals $\Jc$ of $\Lambda$ that contain $\qf^s$ correspond to the proper ideals of $\Lambda/\qf^s\Lambda \cong \Mc_n(\Oc_{{F}}/\qf^s)$,  we see that the ideals $\Jc$ of $\Lambda$ that contain $\qf^s$ are precisely the ideals $\Jc_t$ of the statement of the proposition, and the quotient $\Lambda/\Jc_t \cong \Mc_n(\Oc_F/\qf^t)$.

\end{proof}

\begin{example}
Consider again the cyclic division algebra $\Dc=(\QQ(\imath,\sqrt{5})/\QQ(\imath),\sigma,\imath)$, with 
\[
\Lambda = \ZZ[\imath,(1+\sqrt{5})/2]\oplus\ZZ[\imath,(1+\sqrt{5})/2]z,~z^2=\imath
\] 
of Example \ref{ex:q5}, and take $\Ic = (1+\imath)^2$.  Since $\Oc_F/\langle 1+\imath\rangle^2 \cong \ZZ[\imath]/\langle1+\imath\rangle^2 \cong\ZZ[x]/\langle x^2+1, (x+1)^2 \rangle \cong \FF_2[x]/\langle x^2+1\rangle \cong\FF_2[\imath]$ (where  the last but one isomorphism arises because the ideal $\langle x^2+1, (x+1)^2 \rangle$ also contains $2$ as can be readily seen), we get that
\[
\Lambda/\Ic\Lambda \simeq \Mc_2(\FF_2[i]).
\]
\end{example}

More generally, from a coding perspective, being able to consider quotients of $\qf^s$, $s>1$ is of interest, since it increases the lower bound on the minimum determinant. This can be easily seen for example if $\qf=(\alpha)$ is principal, since then, similarly as shown in (\ref{eq:minprin}), we have 
\begin{equation}\label{eq:minprins}
\Delta_{min} \geq \min_{0\neq X_i} |\det(X_i)|^2 \min\left( d_H(\bar{\Cc})^2 , |\alpha|^{2sn}\right),
\end{equation}
where the minimum Hamming distance of $\bar{\Cc}$ can be increased, especially when the length $L$ of the codeword is large, however $|\alpha|^{2sn}$ is fixed once $s$ and $n$ are given. Thus, for a chosen $n$, having the freedom to increase $s$ provides coding benefit in terms of the minimum determinant of $\Cc$.

%
%
%

\section{The Split Case: $\Ic=\qf$, $g>1$, $e=1$, $f=n/g$}
\label{sec:noram}

We consider now the case where $\Ic=\qf$ factors as $\qf\Oc_K=\Qf_1\Qf_2\cdots \Qf_g$ in $\MCO_K$ for $g>1$, so that $f = n/g$ for each extension $\Qf_i$.  Writing $\overline{K}$ for $\Oc_K/\qf\Oc_K$ and $\Kibar$ for $\MCO_K/\Qf_i $, we get by the Chinese Remainder theorem {and the comaximality of the $\Qf_i$} that 
\begin{eqnarray*} \overline{K} &\cong& \K1bar\times\cdots \times \Krbar\\
k+ \qf\MCO_K &\mapsto& (k+\Qf_1, \cdots, k+\Qf_g).
\end{eqnarray*}

{We use this isomorphism to transfer the action of $G$ on $\overline{K}$ to an action on $\K1bar\times\cdots \times \Krbar$. Thus, if the preimage of $(k_1+\Qf_1, \dots, k_g + \Qf_g)$ is some $k+\qf\MCO_K$, then $\sigma(k_1+\Qf_1, \dots, k_g + \Qf_g)$ is defined to be the image of $\sigma(k+\qf\MCO_K)$, i.e, the element $(\sigma(k) +\Qf_1, \dots, \sigma(k)+\Qf_g)$.  }

{Note that since $\Qf_i\cap \MCO_F = \qf$ for all $i$, $\Fbar = \MCO_F/\qf$ sits inside each $\Kibar$ via $f+\qf \mapsto f + \Qf_i$, and hence inside the direct product via $f+\qf \mapsto (f + \Qf_1,\dots, f+ \Qf_g)$. 
We have that
\begin{equation}\label{eq:Lambda}
\Lambda/\Ic\Lambda \cong \oplus_{j=0}^{n-1} (\MCO_K/\qf\MCO_K)z^j \cong \oplus_{j=0}^{n-1} (\K1bar\times\cdots \times \Krbar)z^i
\end{equation}
where $z(k_1+\Qf_1, \dots, k_g + \Qf_g) = \sigma(k_1+\Qf_1, \dots, k_g + \Qf_g)z$ and {$z^n = (u+\qf, \dots, u+\qf)$} 
}

Let us discuss {the Galois} action further. 

\begin{lemma}
Using the above notations, we have that
\begin{enumerate}
\item
$\Kibar/\Fbar$ is a cyclic finite field extension with Galois group $\langle \sigma^g \rangle$, 
and $\K1bar \cong \Kibar$ (for all $i=2,\dots, g$).
\item
After reordering the primes  $\Qf_1,\ldots,\Qf_g$ if necessary, the action of $\sigma^j$ on $(k,0,\ldots,0)$ in $\K1bar\times\cdots \times \Krbar$ is to send it to $(0,\ldots, \sigma^j(k),\ldots,0)$, where $\sigma^j(k)$ is in the $j+1$ position, and where the position is understood modulo $g$.
\end{enumerate}
\end{lemma}
\begin{proof}
Recall that $G$ acts transitively on the set of prime ideals $\Qf_1, \dots, \Qf_g$.  Since the orbit of any $\Qf_i$ has size $g$, it follows that the stabilizer of $\Qf_i$ is of size $n/g$, and is thus the subgroup generated by $\sigma^g$. Hence, the subgroup generated by $\sigma^g$ yields an induced action on each $\Kibar=\MCO_K/\Qf_i$. 
{As mentioned above, $\Fbar = \MCO_F/\qf$ sits inside each $\Kibar$ and $\Kibar/\Fbar$ is in fact a cyclic finite field extension of degree $f=n/g$}, with Galois group isomorphic to the subgroup generated by $\sigma^g$. (All this is standard, and can be found, for instance, in the discussions in \cite[Chapter III, \S 1]{Jan}.)

Write $x$ for the preimage in $\Kbar$ of an element of the form $(k,0,\dots,0)$ of $\K1bar\times\cdots \times \Krbar$. Thus, $x\equiv k$ (mod $\Qf_1$), and $x\equiv 0$ (mod $\Qf_j$, $j\neq 1$).  It follows that $\sigma(x) \equiv \sigma(k)$ (mod $\sigma(\Qf_1)$), and $\sigma(x)\equiv 0$ (mod $\sigma(\Qf_j)$, $j\neq 1$), where $\sigma(\Qf_1)$ is some $\Qf_i$, $i\neq 1$. Arranging the $\Qf_i$ so that successively, $\sigma$ takes $\Qf_1$ to $\Qf_2$, $\Qf_2$ to $\Qf_3$ and so on until $\sigma^g$ takes $\Qf_1$ back to $\Qf_1$, we find that $\sigma$ takes $(k,0,\dots, 0)$ to $(0,\sigma(k),0,\dots,0)$ and then to $(0,0, {\sigma^2}(k),0,\dots,0)$ and so on, until $\sigma^g$ takes $(k,0,\dots, 0)$ to $(\sigma^g(k),0,\dots, 0)$. Another application of $\sigma$ sends this element to $(0,\sigma^{g+1}(k),0,\dots,0)$,  and so on, until $\sigma^{2g}$ takes $(k,0,\dots, 0)$ to $(\sigma^{2g}(k),0,\dots, 0)$, etc.

{After this reordering, $\sigma^{i-1}$ sends $\Qf_1$ onto $\Qf_i$, and hence induces an isomorphism $\K1bar\cong \MCO_K/\Qf_i \cong \MCO_K/\Qf_i\cong\Kibar$ (for all $i=2,\dots, g$).  This is an $\Fbar$-isomorphism.}

\end{proof}

We have again two cases to consider: $\ubar = u+\qf \neq 0$ and $\ubar = 0$. We start with the former.

\subsection{The case $\ubar \neq 0$}

\begin{proposition}  \label{prop:splitting_I_is_just_q}
Suppose that $\Ic=\qf$ is a prime in $\Oc_F$, such that $\qf\Oc_K = \Qf_1\Qf_2\cdots \Qf_g$ in $K$, with $\ubar \neq 0$ in $\Fbar$. Then $\Lambda/\Ic\Lambda \cong \Mc_n(\Fbar)$.
The only proper two-sided ideal $\Jc$ of $\Lambda$ containing $\Ic$ is $\Ic\Lambda=\oplus_{j=0}^{n-1}\qf\Oc_Kz^j$.
\end{proposition}

\begin{proof} 
Recall from (\ref{eq:Lambda}) that 
$$\Lambda/\Ic\Lambda \cong \oplus_{j=0}^{n-1} (\K1bar\times\cdots \times \Krbar)z^j$$ 
where $z(k_1+\Qf_1, \cdots, k_g + \Qf_g) = \sigma(k_1+\Qf_1, \dots, k_g + \Qf_g)z$, 
{$z^n = (u+\qf, \dots, u+\qf)$}, and where $\sigma$ acts on $\K1bar\times\cdots \times \Krbar$ as described.  

We claim first that $$\Lambda/\Ic\Lambda \cong \oplus_{j=0}^{n-1} (\K1bar\times\cdots \times \Krbar)y^j$$ where $y(k_1+\Qf_1, \cdots, k_g + \Qf_g) = \sigma(k_1+\Qf_1, \dots, k_g + \Qf_g)y$, and importantly, {$y^n = (1+\qf, \dots, 1+\qf)$}, i.e., $y^n = 1_{\Fbar}$.  To see this, note that the norm map $N_{\K1bar/\Fbar}$ from $\K1bar$ to $\Fbar$ is surjective since this is an extension of finite fields. Hence, given that $u+\qf \neq 0_{\Fbar}$, {it is invertible, and recalling from the above lemma the Galois group of $\K1bar/\Fbar$}, there exists $k+\Qf_1\in \K1bar$ such that 
$$N_{\K1bar/\Fbar}(k+\Qf_1) =(k+\Qf_1)\sigma^g(k+\Qf_1)\dots \sigma^{g(f-1)}(k+\Qf_1) = u^{-1}+\qf.$$
In particular, this means that 
$$k\sigma^g(k)\cdots \sigma^{g(f-1)}(k) \equiv u^{-1} (\text{mod}\ \Qf_1)$$ so applying $\sigma^{i-1}$ for any $i$ and noting that $u^{-1}\in\Oc_F$ is fixed by $\sigma$, we find 
$$\sigma^{i-1}\left(k\sigma^g(k)\cdots \sigma^{g(f-1)}(k) \right) \equiv u^{-1} (\text{mod}\ \Qf_i).$$
 Now consider the element $y = wz$, where $w = (k+\Qf_1, 1+\Qf_2,\dots,1+\Qf_g)$.  It is clear that $yx = \sigma(x)y$ for all $x\in \K1bar\times\cdots \times \Krbar$. Moreover, $y^n = w \sigma(w) \cdots \sigma^{g-1}(w) \sigma^g(w) \cdots \sigma^{n-1}(w) (u+\qf)$. Now, {recalling the action of $\sigma$ from the previous lemma}, we have
\begin{eqnarray*}  
\sigma(w) &=& (1+\Qf_1, \sigma(k)+\Qf_2,\dots,1+\Qf_g),\\ 
\sigma^{g-1}(w) &=& (1+\Qf_1, 1+\Qf_2,\dots,\sigma^{g-1}(k)+\Qf_g),\\ 
\sigma^g(w)&=& (\sigma^{g}(k)+\Qf_1, 1+\Qf_2,\dots,1+\Qf_g),
\end{eqnarray*}
 and so on. Multiplying, we find 
{
\begin{eqnarray*}
y^n &=&(N_{\K1bar/\Fbar}(k+\Qf_1),\sigma(N_{\K1bar/\Fbar}(k+\Qf_2)),\dots,\sigma^{g-1}(N_{\K1bar/\Fbar}(k+\Qf_g)))(u+\qf),\\
    &=&(u^{-1}+\qf, \dots,u^{-1}+\qf ) (u+\qf) = (u^{-1}+\qf)(u+\qf) = 1_{\Fbar}
\end{eqnarray*}
using the above relations.}
Since $y^i = w\sigma(w)\cdots\sigma^{i-1}(w)z^i$, the $\Fbar$-subspace $(\K1bar\times\cdots \times \Krbar) z^i$ equals $(\K1bar\times\cdots \times \Krbar) y^i$. It is clear now that using $y$ instead of $z$, we may write $\Lambda/\Ic\Lambda \cong \oplus_{j=0}^{n-1} (\K1bar\times\cdots \times \Krbar)y^j$.

Now $V=\K1bar\times\cdots \times \Krbar$ is a { $gf = n$} dimensional space over $\Fbar$, and $\K1bar\times\cdots \times \Krbar$ embeds $\Fbar$ isomorphically into $End_{\Fbar}(V)$ via left multiplication: $x\in \K1bar\times\cdots \times \Krbar \mapsto \lambda_x$.  Note that the action of $\sigma$ on $\K1bar\times\cdots \times \Krbar$ is $\Fbar$-linear, so there is an element $T  \in End_{\Fbar}(V)$ corresponding to $\sigma$.  Furthermore
\[
T\lambda_x(y) = \sigma(xy) = \sigma(x)\sigma(y),
\]
so 
\[
T\lambda_x = \lambda_{\sigma(x)}{T} 
\]
and $T^n = 1$.  It follows that the map
{
\begin{eqnarray*} 
\Lambda/\Ic\Lambda&\rightarrow& End_{\Fbar}(V)\\
\sum_{j=0}^{n-1} x_jy^j &\mapsto& \sum_{j=0}^{n-1} \lambda_{x_j}T^j
\end{eqnarray*} } is a well defined ring homomorphism. {By count of $\overline{F}$-dimensions, it is enough to show} that it is injective to prove that $\Lambda/\Ic\Lambda \cong End_{\Fbar}(V)$, and since $End_{\Fbar}(V) \cong \Mc_n(\Fbar)$, the proposition will be proved.

To this end, suppose that $\sum_{j=0}^{n-1} {\lambda_{x_j}T^j = 0}$ in $End_{\Fbar}(V)$.  Abusing notation, let us drop the prefix $\lambda$ and write just {$x_j$ for $\lambda_{x_j}$.  Each $x_j$ is of the form $(x_{j,1},\dots,x_{j,g}) \in \K1bar\times\cdots \times \Krbar$.}  We study the action of $x_j T^j$ on a typical element $(k_1,\dots, k_g)$ of $\K1bar\times\cdots \times \Krbar$.  Writing $j = mg+b$ for $0 \le b < g$ and $0 \le m \le f-1$, we find that $x_j T^{mg+b}$ sends $(k_1,\dots, k_g)$ to $x_j(\sigma^{gm}(k_{g-b+1}), \sigma^{gm}(k_{g-b+2}), \dots, \sigma^{gm}(k_{g-b+g}))$, where the subscripts are taken modulo $g$ so as to lie in $\{1,2,\dots, g\}$.  Examining the entry in the first slot, for instance, of {$\sum_{j=0}^{n-1}x_j T^j$} acting on  $(k_1,\dots, k_g)$, we find it to equal
\begin{eqnarray} 
&(\text{$b=0$ terms})& x_{0,1} k_1 + x_{g,1}\sigma^g(k_1) + \dots + x_{(f-1)g,1}\sigma^{(f-1)g}(k_1) +  
\nonumber\\
 &(\text{$b=1$ terms })&x_{1,1} \sigma(k_g) + x_{g+1,1}\sigma^g(\sigma(k_g))  +\dots + x_{(f-1)g+1,1}\sigma^{(f-1)g}(\sigma(k_g))  \nonumber\\
&&  +\ldots + \nonumber\\
&(\text{$b=g-1$ terms})& x_{g-1,1} \sigma^{g-1}(k_2) + x_{g+g-1,1}\sigma^g(\sigma^{g-1}(k_2)) \nonumber \\
&&+ \dots + x_{(f-1)g+g-1,1}\sigma^{(f-1)g}(\sigma^{g-1}(k_2)). \nonumber\\
&& \label{eqn_for_x}
\end{eqnarray} 
The expression above, which involves the $fg = n$ variables $x_{i,1}$, $i=0,\dots, n-1$, should equal $0$. Similarly considering the entries in the various $j$-th slots ($j=2,\dots, g$) {of $\sum_{j=0}^{n-1}x_j T^j$} acting on  $(k_1,\dots, k_g)$, we get $g-1$ equations involving the remaining $x_{i,j}$.

Now let us successively choose $k=(a_1,0,\dots, 0)$, $k=(a_2, 0, \dots, 0)$, and so on, {up} to $k= (a_f,0,\dots, 0)$, where the $a_i$ form an $\Fbar$-basis for $\K1bar$.  Then the terms in (\ref{eqn_for_x}) coming from rows other than that corresponding to $b=0$ all drop out, and we get the following $f$ equations for the variables $x_{0,1}$, $x_{g,1}$, $\dots$, $x_{(f-1)g,1}$:
\begin{equation} \label{disc_matrix}
\left(\begin{array}{ccccc}
a_1 & \sigma^g(a_1) & \sigma^{2g}(a_1)& \dots & \sigma^{(f-1)g}(a_1) \\
a_2 & \sigma^g(a_2) & \sigma^{2g}(a_2)& \dots & \sigma^{(f-1)g}(a_2) \\
a_3 & \sigma^g(a_3) & \sigma^{2g}(a_3)& \dots & \sigma^{(f-1)g}(a_3) \\
\vdots & \vdots &\vdots &\vdots &\vdots\\
a_{f} & \sigma^g(a_f) & \sigma^{2g}(a_f)& \dots & \sigma^{(f-1)g}(a_f) \\
\end{array}
\right) 
\left(
\begin{array}{c}
x_{0,1}\\
x_{g,1}\\
x_{2g,1}\\
\vdots\\
x_{(f-1)g,1}\\
\end{array}
\right) = 
\left(
\begin{array}{c}
0\\
0\\
0\\
\vdots\\
0\\
\end{array}
\right)
\end{equation}
But the matrix $M$ on the left side of the equation above is nonsingular: this follows from the fact that the product $MM^T$ is the matrix whose $(i,j)$-th entry is the trace (from $\K1bar$ to $\Fbar$) of $a_ia_j$, and since the $a_i$ form a basis of (the separable extension) $\K1bar/\Fbar$,  $det(MM^T) = disc(a_1,\dots, a_f) \neq 0$.  Hence, $x_{0,1}$, $\dots$, $x_{(f-1)g,1}$ must all be zero.

Similarly, we can prove that all remaining variables $x_{i,1}$, $i=0,\dots, n-1$ must be zero, by considering corresponding values of $k$ that are zero in all but one slot.  {Further,} considering the other $g-1$ equations that involve the remaining $x_{i,j}$, {$j=2,\dots, g$, and applying} the same technique, we find that all $x_{i,j}$ must be zero.

{That $\Ic\Lambda$ is the only two-sided ideal that contains $\qf$ follows as in Proposition~\ref{prop:f=n}.}
\end{proof}

\subsection{The case $\ubar = 0$}

We now deal with the case $\ubar = 0_{\Fbar}$, i.e., $u\in \qf$. 
Write $v_i$ for the element $(0,\dots, 0,1,0,\dots,0)$, where the $1$ is in the $i$-th slot.  Also, given a pair $(i,j)$ with $1 \le i \le g$, and $0 \le j < n$, call the set of integer pairs $(p,q)$ such that $i \le p $, and $j+(p-i) \le q < n$, with the identification $(p,q) \sim (p+g, q)$ the \textit{cyclic stairwell} detemined by $(i,j)$. We have the following:
\begin{proposition}\label{prop:second} 
Suppose that $\Ic=\qf$ is a prime in $\Oc_F$ such that $\qf\Oc_K=\Qf_1,\ldots,\Qf_g$ in $K$ and $\ubar = 0$ in $\Fbar$. Then 
\[
\Lambda/\Ic\Lambda  \cong   \oplus_{j=0}^{n-1} (\K1bar\times\cdots \times \Krbar)z^j,
\]
where $z(k_1+\Qf_1, \cdots, k_g + \Qf_g) = \sigma(k_1+\Qf_1, \dots, k_g + \Qf_g)z$, {$z^n = (0+\qf_1, \dots, 0+\qf_g = 0_{\Fbar}$}. Every nonzero ideal of $\Lambda/\Ic\Lambda$ is minimally generated by a set of ``monomials'' $v_{i_1}z^{j_1}, \dots, v_{i_t}z^{j_t} $, $i_1 < i_2 <\cdots < i_t$, with the property that $v_{i_s}z^{j_s}$ does not lie in the cyclic stairwell determined by any of the monomials $v_{i_1}z^{j_1}$, $\dots$, $v_{i_{s-1}}z^{j_{s-1}}$. Thus, the ideals of $\LA$ that contain $\Ic$ are the ideals generated by $\qf$ and monomials $v_{i_1}z^{j_1}, \dots, v_{i_t}z^{j_t} $, $i_1 < i_2 <\cdots < i_t$ with the property described.
\end{proposition}

(We  would like to thank Ken Goodearl and Kenny Brown for pointing out that \textit{every} ideal in the ring $\Lambda/\Ic\Lambda$ is generated by monomials.)

\begin{proof} That 
\[
\Lambda/\Ic\Lambda  \cong   \oplus_{j=0}^{n-1} (\K1bar\times\cdots \times \Krbar)z^j,
\]
subject to the constraints above follows from (\ref{eq:Lambda}).
Let $\Jc$ be a nonzero ideal of $\Lambda/\Ic\Lambda$.  Pick any nonzero element $b = \sum_{j=0}^{n-1} a_j z^j$ in $\Jc$, where each $a_j$ is of the form $(a_{1,j}, \dots, a_{g,j})$, with the $a_{i,j} \in \Kibar$.   Note that $v_1 + \cdots + v_g = 1$.  It follows that $b = v_1b + \cdots + v_g b$.  Now each $v_ib$ is in the ideal generated by $b$, so the ideal generated by all the $v_ib$ is contained in the ideal generated by $b$.  But coupling this with the relation $b = v_1b + \cdots + v_g b$ we find that the ideal generated by $b$ equals the ideal generated by all the $v_ib$.

We will first show that the ideal generated by each $v_i b$ equals the ideal generated by a single monomial $v_i z^j$ for suitable $j$.  It will follow that the ideal generated by $b$ equals the ideal generated by a  suitable set of monomials, from which we can conclude that $\Jc$ itself is generated by some set of monomials.

We have $v_ib = \sum_{j=0}^{n-1} v_i a_j z^j$.  If all $v_i a_j$ are zero, then $v_i b = 0$, and there is nothing to prove.  So assume that $s$ is least such that $v_i a_s \neq 0$.  Using $v_i^2 a_s= v_ia_s$, we write $v_i b = v_i \left( \alpha z^s + a_{s+1}z^{s+1} + \cdots \right)$, where we have written $\alpha$ for $v_i a_s = (0,\dots, a_{i,s}, \dots, 0)$.  Write $\beta$ for $(0,\dots, a_{i,s}^{-1}, \dots, 0)$.  Multiplying, we find  $\beta v_i b = v_i\left(v_i z^s + a'_{s+1}z^{s+1} + \cdots \right)$ for suitable $a'_{s+1}$, $\dots$.  The first summand on the right is $v_i^2 z_s$ which is just $v_i z^s$.  Thus,  $\beta v_i b = 
v_iz^s\left(1+ \sigma^{-s}(a'_{s+1})z + \cdots \right) = v_i z^s (1+c z)$ for suitable $c \in \Lambda/\Ic\Lambda$.

Now note that  $z c = c' z$ for suitable $c' \in \Lambda/\Ic\Lambda$. It follows that $(c z)^n = c z c z \cdots c z = \tilde{c}z^n  = 0$. By a standard trick, we find that  $1+c z$ is invertible: its inverse is $1 -cz +(cz)^2 + \cdots +(-1)^{n-1}(cz)^{n-1}$.  Hence, $\beta v_i b = v_i z^s u$ for a unit $u$, from which we find that $v_i z^s = \beta v_i b u^{-1}$ is in the ideal generated by $v_i b$.  Similarly, we find $\alpha \beta v_i b =v_i^2 b = v_i b =\alpha v_i z^s u$, from which we find that $v_ib$ is in the ideal generated by $v_i z^s$.  Thus, the ideal generated by $v_ib$ equals the ideal generated by $v_i z^s$.  As described above, it follows that $\Jc$ is generated by a set of monomials.

Finally, given $v_iz^j\in \Jc$, we know that $v_i z^{j+k}\in \Jc$, $k = 1, 2, \dots$. Similarly, $zv_iz^j = v_{i+1}z^{j+1}\in \Jc$.  (Note that this relation must considered cyclically, that is, $z v_g z^{j} = v_1 z^{j+1} \in \Jc$.) Proceeding thus,  we find that all monomials in the cyclic stairwell determined by $v_iz^j$ are already in the ideal generated by $v_iz^j$ alone. The statement about the minimal generating set for $\Jc$ follows immediately after noting that there are only a finite number of monomials in $\Lambda/\Ic\Lambda$.
\end{proof}

%
%
%

\section{The {Split} Case: $\Ic=\qf^s$, $g>1$, $e=1$, $f=n/g$}
\label{sec:noram_power_of_a_prime}

We consider in this section the case where $\Ic$ is a \textit{power} of a prime ideal, but where, as in the previous section, the ideal $\qf$ factors as $\qf\Oc_K=\Qf_1\Qf_2\cdots \Qf_g$ in $\MCO_K$ for $g>1$.  Thus, $\Ic = \qf^s$ (for some $s> 1$), and $f = n/g$ for each extension $\Qf_i$. 

We will study in this section the case where $u\not\in \qf$.

Recall that if $\widehat{F}$ is the completion of $F$ at $\qf$, then $K \otimes_F \widehat{F}  \cong \widehat{K^{(1)}} \times \cdots \times \widehat{K^{(g)}}$, where the various $\widehat{K^{(i)}}$ are the completions of $K$ at the extensions $\Qf_i$ of $\qf$. Moreover, $\langle \sigma \rangle$ acts on $\widehat{K^{(1)}} \times \cdots \times \widehat{K^{(g)}}$ (after renumbering the $\widehat{K^{(i)}}$) as follows:  $\sigma$ sends $(0,\dots, 1,\dots, 0)$, where $1$ is in the $i$-th position, to $(0,\dots, 1,\dots, 0)$, where $1$ is in the $i+1$-th position but taken modulo $g$, and $\sigma^g$ acts as a Galois automorphism of $\widehat{K^{(i)}}$  with fixed field precisely $\widehat{F}$.  (This is standard, and is the content of \cite[Chap. III, Theorem 1.2 ]{Jan}, for instance.)

Let $\Vc_{\widehat{K^{(i)}}}$ stand for the valuation ring in $\widehat{K^{(i)}}$ of the $\Qf_i$-adic valuation, and let $\Vc_{\widehat{F}}$ stand for the valuation ring in ${\widehat{F}}$ of the $\qf$-adic valuation.  Consider the ring $\widehat{\Lambda} = \oplus_{j=0}^{n-1} \left(  \Vc_{\widehat{K^{(1)}}} \times \cdots \times \Vc_{\widehat{K^{(g)}}} \right) z^j$, with relations $z(k_1, \dots, k_g) = \sigma(k_1,\dots, k_g)z$ for $(k_1\dots,k_g) \in \Vc_{\widehat{K^{(1)}}} \times \cdots \times \Vc_{\widehat{K^{(g)}}}$, and $z^n = u$, where of course $u\in \Vc_F$ corresponds to the element $(u,\dots, u) \in \Vc_{\widehat{K^{(1)}} }\times \cdots \times \Vc_{\widehat{K^{(g)}}}$.  This is an algebra over the local ring $\Vc_{\widehat{F}}$.  It is easy to see that it is a free module over the ring $\Vc_{\widehat{F}}$, since each $\Vc_{\widehat{K^{(i)}}}$ is  a free module because $\MCO_{\widehat{F}}$ is a p.i.d.  Moreover, continuing to write $\qf$ for $\qf \Vc_{\widehat{F}}$, the quotient $\widehat{\Lambda} /\qf \widehat{\Lambda} $ is precisely $\Mc_n(\Fbar)$ by Proposition \ref{prop:splitting_I_is_just_q}.  (This follows by taking $s=1$ in the isomorphisms $\MCO_F/\qf^s \MCO_F \cong \MCO_{\widehat{F}}/\qf^s \MCO_{\widehat{F}}$ for $s = 1, 2, \dots$.)  Thus, $\widehat{\Lambda}$ is an Azumaya algebra over $\MCO_{\widehat{F}}$, and as previously argued in the proof of Proposition \ref{prop:inert_power_of_prime_u_not_in_q}, this forces $\widehat{\Lambda} \cong \Mc_n(\MCO_{\widehat{F}})$.  This leads to the following:

\begin{proposition} \label{prop:power-of-prime-unramified}
In the situation described above ($\Ic = \qf^s$, $s > 1$, $\qf\Oc_K = \Qf_1\Qf_2\cdots\Qf_g$, $u \not\in\qf$), $\Lambda/\Ic\Lambda \cong \Mc_n(\Oc_F/\qf^s)$. The two-sided proper ideals $\Jc$ of $\Lambda$ such that $\Jc$ contains $\qf^s \Lambda$ are the ideals $\Jc_t = \qf^t \Lambda = \oplus_{j=0}^{n-1} \qf^t\Oc_K z^j$, for $1 \le t \le s$, and their quotients satisfy $\Lambda/\Jc_t \cong \Mc_n(\Oc_F/\qf^t)$ . 

\end{proposition}

\begin{proof} We observe that for any $t \ge 1$, $\MCO_K/\qf^t \MCO_K \cong  \MCO_{K}/\Qf_1^t \MCO_{K} \times \cdots \times \MCO_{K}/\Qf_g^t \MCO_{K} \cong \Vc_{\widehat{K^{(1)}}}/\Qf^t \Vc_{\widehat{K^{(1)}}} \times \cdots \times \Vc_{\widehat{K^{(g)}}}/\Qf^t \Vc_{\widehat{K^{(g)}}}$.   It follows from this that
 $\Lambda/\qf^t\Lambda \cong \widehat{\Lambda}/\qf^t\widehat{\Lambda} $ for $t\ge 1$.  Since $\widehat{\Lambda} \cong \Mc_n(\MCO_{\widehat{F}})$, we find $\Lambda/\qf^s\Lambda \cong \Mc_n(\MCO_F/\qf^s)$, just as in
the proof of Proposition \ref{prop:inert_power_of_prime_u_not_in_q}.  The remaining statements about the ideals are proved just as in Proposition   \ref{prop:inert_power_of_prime_u_not_in_q}.

\end{proof}

%
%
\section{Further Coding Remarks and Applications}
\label{sec:wiretap}

Quotients $\Lambda/\Jc$ of the natural order $\Lambda$ of a cyclic division algebra $\Dc$ by a two-sided ideal $\Jc$ appear in the context of space-time coded modulation. From (\ref{eq:dhbound}), the main design criterion for a space-time code $\Cc$ is to maximize the minimum determinant $\Delta_{min}$ which is lower bounded by
\[
\Delta_{min} \geq \min\left( d_H(\bar{\Cc})^2 \min_{0\neq x_i} |\det(X_i)|^2, \min_{0\neq x_i \in \Jc} |\det(X_i)|^2\right).
\]
There are three factors that influence this lower bound:
\begin{enumerate}
\item
$\min_{0\neq x_i} |\det(X_i)|^2$ which depends on the inner space-time code designed on the cyclic division algebra $\Dc$,
\item
the Hamming distance $d_H(\bar{\Cc})$, 
\item
and the index of $\Jc$ in $\Lambda$.
\end{enumerate}
The different quotient rings seen in the previous sections illustrate well some of the trade-offs of this code design: if $\Lambda/\Jc\cong \Mc_n(\bar{F})$, the lower bound on $\Delta_{min}$ is good, and the coding problem is that of finding codes with suitable Hamming distance over $\Mc_n(\bar{F})$. If $\Lambda/\Jc\cong \Mc_n(\Oc_K/\qf^s)$, the lower bound is much better, but the matrix ring considered becomes more complex. Important from a coding point of view is the number of symbols encoded. We observed that when $\Lambda/\Jc \cong \bar{K}$, the lower bound is smaller, but more symbols are encoded.

Coset encoding in general is used in a variety of coding scenarios. We will next briefly mention its application in the context of wiretap codes. Recall that a wiretap 
channel \cite{Wyner} is a broadcast channel where the transmitter, Alice, sends messages 
to Bob, a legitimate receiver, and Eve, an eavesdropper. A wiretap code should ensure 
that transmission between Alice and Bob is reliable, and the amount of information that 
Eve receives is negligeable. Encoding a wiretap code involves including bits of randomness 
to increase Eve's confusion, and this is done via coset encoding with $G=\FF_2^n$ and $H=C$ an $(n,n-k)$ code. There 
are then $2^k$ cosets of $C$, $k$ bits of information to label them, and the $n-k$ uncoded 
bits are random bits, used to pick one codeword at random with the chosen coset. The same 
holds for wiretap lattice codes, as proposed in \cite{BOjournal}, for MIMO (multiple input multiple output) channels , where the message 
transmitted by Alice is now an $n\times n$ space-time code, a matrix with complex coefficients, which can be seen as a lattice $\Lambda_b$ in $\mathbb{C}^{n^2}$, with $\Lambda_e$ a sublattice of $\Lambda_b$.
Suppose that $|\Lambda_b/\Lambda_e|=2^k$, then a message of $k$ bits will label the different cosets of $\Lambda_e$, while random bits are used to pick a point at random in the coset, that will be the actual transmitted point.
This again addresses the problem of coset encoding in cyclic division algebras, using a different code design from that studied in this paper.

As described in \S \ref{subsecn:contrib}, the results of this paper provide a framework for coset coding in division algebras by identifying a large class of quotient rings of the natural order $\LA$. To take full advantage of this framework, we need to study codes over the various quotient rings that we have identified, and determine best ways to lift them to a code over $\LA$.  This will be the focus of future work.

%
%

\end{document}